\documentclass[]{interact}

\usepackage{epstopdf}% To incorporate .eps illustrations using PDFLaTeX, etc.
\usepackage[caption=false]{subfig}% Support for small, `sub' figures and tables
\usepackage{tikz}
\usepackage{booktabs}
\usepackage{color}
\usepackage{enumitem}
\usepackage{verbatim}
\usetikzlibrary{positioning}
\usetikzlibrary{shapes.geometric}
\usepackage{setspace}

\usepackage{algorithm}
\usepackage{algpseudocode} 
\usepackage{multirow} 
\usepackage{hyperref}

\usepackage[numbers,sort&compress]{natbib}% Citation support using natbib.sty
\bibpunct[, ]{[}{]}{,}{n}{,}{,}% Citation support using natbib.sty
\theoremstyle{plain}% Theorem-like structures provided by amsthm.sty
\newtheorem{theorem}{Theorem}[section]

\newtheorem{proposition}[theorem]{Proposition}

\theoremstyle{definition}
\newtheorem{definition}[theorem]{Definition}
\newtheorem{example}[theorem]{Example}

\theoremstyle{remark}
\newtheorem{remark}{Remark}

\begin{document}

\articletype{RESEARCH ARTICLE}% Specify the article type or omit as appropriate

\title{SL-vine and $\Delta$-vine copula models}

\author{
\name{Tanabodee Monggon\textsuperscript{a}, Songkiat Sumetkijakan\textsuperscript{a}, Tippawan Santiwipanont\textsuperscript{a} and Monchai Kooakachai\textsuperscript{a}\thanks{CONTACT T. Santiwipanont. Email: Tippawan.S@chula.ac.th}}
\affil{\textsuperscript{a}Department of Mathematics and Computer Science, Faculty of Science, Chulalongkorn University, Bangkok, Thailand}}

\maketitle
\begin{abstract}
Traditional multivariate copulas often fail to capture complex dependence structures and may impose restrictive assumptions that limit their applicability to real-world data. To address these issues, vine copulas provide a flexible framework by constructing multivariate models from bivariate copulas. The three main types are R-, D-, and C-vines. To simplify computations, the ``simplifying assumption" is often applied, assuming conditional copulas are independent of conditioning variables. This results in some variables exerting minimal influence on the conditional copulas. Therefore, this study explores two subclasses of R-vines with node degrees capped at three, assessing its performance against existing models. The statistical results demonstrate that the proposed degree-constrained vine structures can serve as effective alternatives to existing vine copula models.
\end{abstract}


\section{Introduction}
\label{sec1}
%%%%%%%%%%%%%%%%%%%%%%%%%%%%%%%%

Copulas gained its popularity in the 80's. Sklar’s theorem establishes that any \(n\)-dimensional joint distribution can be decomposed into its marginal distributions and a copula function that uniquely characterizes the underlying dependence structure. Beyond academia, copulas gained significant public attention during the financial crisis, where they were widely criticized for their role in risk modeling (see, for example,  \cite{jones2009formula}, \textit{The Formula That Felled Wall Street}). In particular, the Gaussian copula, which was extensively used in financial applications, is unable to capture heavy-tailed dependence and therefore fails to adequately model extreme market events such as those observed during the crisis.

To obtain greater modeling flexibility, multivariate copulas may be constructed using the pair-copula construction (PCC) introduced by \cite{aas2009pair}, building upon the earlier work of  \cite{joe1996families}. This framework decomposes a joint distribution into conditional bivariate copulas, allowing for more marginals flexible dependence structures. Owing to the large number of possible PCC configurations, \cite{bedford2001probability} introduced regular vines as a systematic graphical framework for organizing these constructions. Regular vines consist of connected tree sequences that describe both the marginal distributions and their dependence relationships.

Various possible configurations of PCC, giving rise to the class of regular vine specifications. Among the most widely studied structures are the R-vine, D-vine, and C-vine copulas. An R-vine is represented by a sequence of connected trees in which the edges of one tree become the nodes of the next. Within this framework, nodes of the first tree correspond to random variables, while edges represent pair-copula dependencies between variables. The D-vine constitutes a restricted form of the R-vine where each node is connected to at most two others. In contrast, the C-vine is organized around a central node, known as the root, which is directly linked to all remaining variables. 

To improve the tractability of vine copula models, the \textit{simplifying assumption} is commonly imposed, requiring conditional copulas to be independent of the conditioning variables. In a C-vine, this means that the conditional copulas are assumed to be unaffected by the root variable, whereas in highly connected vine structures, many conditioning variables are likewise assumed to have negligible influence on the associated conditional copulas. In practice, however, dependence structures often lie between these two extremes, with only a limited number of variables exerting moderate effects on neighboring conditional copulas. This motivates the study of R-vine subclasses with restricted node degrees. In particular, we focus on vine structures in which each node has degree at most three, starting with a structure consisting of stems and leaves that exhibits a hybrid characteristic of D-vines and C-vines, referred to as the SL-vine. Then a generalization of the SL-vine is introduced, which is called the \(\Delta\)-vine.

This paper proposes modified vine copula structures with degree constraints for modeling multivariate dependence and evaluates their performance relative to existing vine copula models. Section~\ref{sec2} introduces vine copula specifications and the role of regular vines in organizing pair-copula constructions. Sections~\ref{sec3}--\ref{sec4} develop subclasses of R-vines in which each node has degree at most three and present algorithms for constructing the corresponding tree sequences. In particular, these sections introduce the SL-vine and the \(\Delta\)-vine, respectively, which constitute the main contribution of this work.
Section~\ref{sec5} provides examples of simulating from SL-vine and $\Delta$-vine copulas. Section~\ref{sec6} compares the proposed models with existing vine structures using the abalone, capital-market, and Wisconsin breast cancer datasets. Finally, Section~\ref{sec7} concludes the paper and outlines possible directions for future research.

%%%%%%%%%%%%%%%%%%%%%%%%%%%%%%%%
\section{Introduction to vine copulas}\label{sec2}
%%%%%%%%%%%%%%%%%%%%%%%%%%%%%%%%

\subsection{Copula}
A function \(C:[0,1]^d \to [0,1]\) is called a \(d\)-copula if it is the joint distribution function of a random vector \(\mathbf{U}=(U_1,\dots,U_d)\), where each component \(U_k\) follows the standard uniform distribution on \([0,1]\), that is,
\[
\mathbb{P}(U_k \leq u_k)=u_k, \qquad u_k \in [0,1], \text{ and } k=1,\dots,d.
\]

The importance of copulas lies in their ability to separate marginal behavior from the dependence structure of multivariate distributions. This property is formalized by Sklar’s theorem \cite{sklar1959fonctions}, which states that for continuous random variables \(X_1,\dots,X_d\) with joint distribution function \(F\) and marginal distribution functions \(F_1,\dots,F_d\), there exists a unique copula \(C\) such that
\begin{align}\label{monggon:eq1}
    F(x_1,\dots,x_d)
    =
    C\bigl(F_1(x_1),\dots,F_d(x_d)\bigr),
    \qquad x_1,\dots,x_d \in \mathbb{R}.
\end{align}
The copula \(C\) is referred to as the copula associated with \(X_1,\dots,X_d\).

If \(X_1,\dots,X_d\) are continuous random variables with joint density function \(f\) and marginal densities \(f_1,\dots,f_d\), then
\begin{align}\label{monggon:eq2}
    f(x_1,\dots,x_d)
    =
    c\bigl(F_1(x_1),\dots,F_d(x_d)\bigr)
    \prod_{i=1}^{d} f_i(x_i),
    \qquad x_1,\dots,x_d \in \mathbb{R},
\end{align}
where \(c = \partial_1 \cdots \partial_d C\) denotes the density function of the copula \(C\).

%%%%%%%%%%%%%%%%%%%%%%%%%%%%%%%%
\subsection{Vine copula}
Although representations \eqref{monggon:eq1} and \eqref{monggon:eq2} provide a general framework for multivariate dependence modeling, they may involve highly complex multivariate copula functions and densities. \cite{joe1996families} showed that the density of a \(d\)-dimensional copula can be decomposed into \(d(d-1)/2\) bivariate copula densities. Since such decompositions are generally not unique, considerable flexibility arises in selecting an appropriate factorization. To systematically characterize and organize these factorizations, \cite{bedford2001probability} introduced the graphical framework of \textit{regular vines} (R-vines). An R-vine on \(d\) variables is represented by a sequence of \(d-1\) connected trees, $\mathcal{V} = (T_1,\dots,T_{d-1})$, where the edges of \(T_k\) become the nodes of \(T_{k+1}\). In this framework, nodes correspond to random variables, while edges represent bivariate copulas between pairs of variables. The collection of pair-copula families is denoted by \(\mathcal{B}(\mathcal{V})\), and the associated parameter set is written as \(\Theta(\mathcal{B}(\mathcal{V}))\). More formally, a vine tree sequence, or vine structure, \(\mathcal{V}=(T_1,\dots,T_{d-1})\), with \(T_k=(N_k,E_k)\) for \(1 \leq k \leq d-1\), must satisfy the following conditions.
\begin{definition}\label{R-vine}
A sequence of trees \(\mathcal{V} = (T_1, \ldots, T_{d-1})\) defined on a node set \(V_1\) containing \(d\) elements is called an \emph{R-vine} if it satisfies the following conditions:
\begin{itemize}
    \item[(i)] The first tree $T_1$ has a vertex set $N_1 =V_1$ and an edge set $E_1$.
    \item[(ii)] For each $k = 2, \ldots, d-1$, the tree $T_k$ is formed over the node set $N_k=E_{k-1}$.
    \item[(iii)]For each $k = 2, \ldots, d-1$, if two nodes of the tree $T_k$ are connected by an edge, then the two nodes, which are edges of $T_{k-1}$, have a common node \textbf{(proximity condition)}.
\end{itemize}
\end{definition}

 According to the proximity condition, an edge can connect two nodes, denoted as $\{a_1, a_2\}$ and $\{b_1, b_2\}$, only when their associated edges in the preceding tree have exactly one vertex in common.
\begin{remark}\label{rem: R-vine}
    According to the proximity condition, an edge can connect two nodes, denoted as $\{a_1, a_2\}$ and $\{b_1, b_2\}$, only when their associated edges in the preceding tree have exactly one vertex in common. Nevertheless, connecting every such pair with an edge is not permissible, since this action could create cycles and thus breach the fundamental property of a tree. This issue particularly arises when a shared node in the earlier tree is incident to three or more edges, necessitating a selective omission of connections to preserve acyclicity.
\end{remark}

A vine tree sequence consisting entirely of path graphs is called a \emph{drawable vine} (D-vine), whereas a sequence composed entirely of star graphs is referred to as a \emph{canonical vine} (C-vine), following \cite{bedford2001probability}. To illustrate the pair-copula construction, consider the decomposition of a three-dimensional density \(f(x_1,x_2,x_3)\):
\begin{align}\label{monggon:eq2.3}
f(x_1, x_2, x_3)
&=
c_{1,3 \mid 2}
\bigl(
F_{1|2}(x_1 \mid x_2),
F_{3|2}(x_3 \mid x_2)
\mid x_2
\bigr)
\,
c_{2,3}
\bigl(
F_2(x_2),
F_3(x_3)
\bigr)
\notag\\
&\quad\times
c_{1,2}
\bigl(
F_1(x_1),
F_2(x_2)
\bigr)
f_1(x_1)f_2(x_2)f_3(x_3).
\end{align}
For an index set \(D \subset \{1,\ldots,d\}\) and an index \(j \notin D\), the functions \(F_{j|D}\) and \(f_{j|D}\) denote the conditional distribution and conditional density of \(X_j\) given \(\mathbf{X}_D\), respectively, where $\mathbf{X}_D = (X_i)_{i \in D}$ with observed value $\mathbf{x}_D = (x_i)_{i \in D}$. The copula density $c_{1,3 \mid 2}(\cdot,\cdot \mid x_2)$ corresponds to the dependence structure of the conditional distribution of \((X_1,X_3)\) given \(X_2=x_2\). The decomposition in \eqref{monggon:eq2.3} can be represented by the D-vine tree structure shown in Figure~\ref{monggon:fig2.1}.

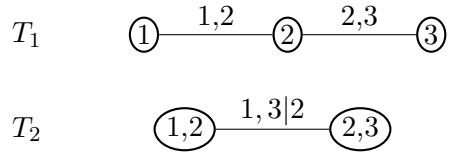
\begin{figure}[htbp]
\centering
\begin{tikzpicture}[every node/.style={thick, inner sep=1pt, font=\normalsize}, node distance=1.5 cm]

    % T1 layer
    \node (T1) at (-2,0) {$T_1$};
    
    \node[ellipse, draw] (2) [right=3cm of T1] {2};
    \node[ellipse, draw] (1) [left=of 2] {1};
    \node[ellipse, draw] (3) [right=of 2] {3};

    \draw (1) -- node[above ] {{1,2}} (2);
    \draw (2) -- node[above ] {{2,3}}(3);
    
    % T2 layer
    \node (T2) at (-2,-1.25) {$T_2$};
    \node[ellipse, draw] (23) [right=3.75 cm of T2] {2,3};
    \node[ellipse, draw] (12) [left=of 23] {1,2};
    \draw (12) -- node[above ] {{$1,3 | 2$}} (23);
\end{tikzpicture}
\caption{D-vine tree structure illustrating the construction of the three pair copulas $c_{23} - c_{12} - c_{13 \mid 2}$}
    \label{monggon:fig2.1}
\end{figure}
By Definition~\ref{R-vine}, the D-vine corresponding to \eqref{monggon:eq2.3} consists of the node set $N_1=\{1,2,3\}$
and edge set $E_1=\{\{1,2\},\{2,3\}\}$ in the first tree \(T_1\). The second tree \(T_2\) then has node set $N_2=\{\{1,2\},\{2,3\}\}$
and edge set $E_2=\bigl\{\{\{1,2\},\{2,3\}\}\bigr\}$. However, due to the proximity condition in Definition~\ref{R-vine}, the representation of edges in higher-order trees becomes increasingly intricate. To simplify the notation, \cite{beare2015vine} introduce the \emph{complete union} associated with an edge \(e=\{a,b\}\in E_k\), defined by
\begin{align*}
\mathcal{U}_e=
\begin{cases}
e,
& \text{if } k=1,\\[0.2cm]
\bigl\{
i \in N_1 :
\exists (e_1,\dots,e_{k-1}) \in E_1 \times \cdots \times E_{k-1}
\text{ such that }
i \in e_1 \in \cdots \in e
\bigr\},
& \text{if } k \geq 2.
\end{cases}
\end{align*}
The set $D_e := \mathcal{U}_a \cap \mathcal{U}_b$ is called the \emph{conditioning set} of \(e\). The corresponding \emph{conditioned sets} are defined by
\[
\mathcal{C}_{e,a}
=
\mathcal{U}_a \setminus D_e
=: a_e,
\text{ and }
\mathcal{C}_{e,b}
=
\mathcal{U}_b \setminus D_e
=: b_e.
\]

Moreover, \cite{kurowicka2006uncertainty} showed that the sets \( \mathcal{C}_{e,a} \) and \( \mathcal{C}_{e,b} \) are singletons. Hence, the notation for an edge \( e = (\mathcal{C}_{e,a}, \mathcal{C}_{e,b} \mid D_e) \) is commonly abbreviated as \( e = (a_e, b_e \mid D_e) \). For example, consider the edge \( e = \{ a = \{1,2\}, b = \{2,3\} \} \), where \( \mathcal{U}_a = \{1,2\} \) and \( \mathcal{U}_b = \{2,3\} \), yielding \( D_e = \{2\} \). It follows that \( \mathcal{C}_{e,a} = \{1\} \) and \( \mathcal{C}_{e,b} = \{3\} \), so that the edge can be written as \( e = (1,3 \mid 2) \). In general, for \(k \geq 2\), each edge in the vine sequence corresponds to a bivariate copula family denoted by \( c_{a_{e}, b_{e} | D_e} \in \mathcal{B}(\mathcal{V}) \), with associated parameter set \( \theta_{a_{e}, b_{e} | D_e} \in \Theta(\mathcal{B}(\mathcal{V})) \). The copula density \(c_{a_e, b_e | D_e}\) represents the dependence structure of the conditional distribution of \((X_{a_e}, X_{b_e})\) given \(\mathbf{X}_{D_e}\). Equivalently, \(c_{a_e, b_e | D_e}\) is the copula density of the random vector
\[
(U_{a_e | D_e}, U_{b_e | D_e}) = \big((F_{a_e | D_e} (X_{a_e} | \mathbf{X}_{D_e}), F_{b_e | D_e} (X_{b_e} | \mathbf{X}_{D_e})) \mid \mathbf{X}_{D_e}\big).
\]

The graphical structure of regular vines provides a natural framework for decomposing multivariate density functions, as demonstrated by \cite{bedford2001probability}. In particular, the pair copulas associated with the vine tree sequence yield the following decomposition:
\begin{align}\label{monggon:eq3}
        f(x_1, \dots, x_d) &=  \left[\prod_{j=1}^{d-1} \prod_{e \in E_k} c_{a_e, b_e | D_e}\left( F_{a_e | D_e} (x_{a_e} | \mathbf{x}_{D_e}), F_{b_e | D_e} (x_{b_e} | \mathbf{x}_{D_e}) \mid \mathbf{x}_{D_e} \right)\right]\notag\\
        &\quad\times\left[\prod_{k=1}^d f_k(x_k)\right].
\end{align}
Equation~(\ref{monggon:eq3}) shows that any regular vine copula density can be represented through a product of bivariate copula densities and marginal densities. However, each pair copula \( c_{a_e, b_e | D_e} \) depends on the specific value of the conditioning vector \( \mathbf{x}_{D_e} \). Consequently, the conditional dependence structure between \( a_e \) and \( b_e \) may vary with different realizations of \( \mathbf{x}_{D_e} \).

For tractability, vine copula models are commonly constructed under the \textit{simplifying assumption}, which states that conditional copulas do not depend on the values of the conditioning variables; see \cite{HOBAEKHAFF20101296, STOBER2013101} for further discussion. More precisely, a distribution \(F\) satisfies the simplifying assumption with respect to the vine structure \(\mathcal{V}\) if equation~(\ref{monggon:eq3}) holds with
\begin{align*}
        c_{a_e, b_e | D_e}( \cdot \mid \mathbf{x}_{D_e} )=c_{a_e, b_e | D_e}( \cdot)\text{ for all } \mathbf{x}_{D_e}\in\mathbb{R}^d\text{ and } e\in\mathcal{V}.
\end{align*} 
Under this assumption, equation~(\ref{monggon:eq3}) reduces to
\begin{align}\label{monggon:eq4}
        f(x_1, \dots, x_d) =  \left[\prod_{j=1}^{d-1} \prod_{e \in E_j} c_{a_e, b_e | D_e}\right]\times\left[\prod_{k=1}^d f_k(x_k)\right],
\end{align}
where $c_{a_e, b_e | D_e}:=c_{a_e, b_e | D_e}\left( F_{a_e | D_e} (x_{a_e} | \mathbf{x}_{D_e}), F_{b_e | D_e} (x_{b_e} | \mathbf{x}_{D_e}) \right)$. Within this framework, the resulting multivariate density construction is referred to as a pair-copula construction.

Vine copulas are designed to model the dependence structure of multivariate random variables. When several admissible connected tree structures are available, it is natural to prioritize edges corresponding to stronger dependencies. This principle is incorporated into the construction of vine tree sequences.
To quantify dependence strength, \textbf{measures of concordance} are commonly employed, among which \textbf{Kendall’s Tau} is one of the most widely used. Let \((X,Y)\) be a continuous random vector. Kendall’s Tau is defined by
\[
\tau_{(X, Y)} \equiv \mathbb{P}((X_1 - X_2)(Y_1 - Y_2) > 0) - \mathbb{P}((X_1 - X_2)(Y_1 - Y_2) < 0),
\]
where \((X_1,Y_1)\) and \((X_2,Y_2)\) are independent random vectors having the same joint distribution \(F\) as \((X,Y)\).

%%%%%%%%%%%%%%%%%%%%%%%%%%%%%%%%    
\subsection{Specification of a vine copula model}
A complete vine copula specification is represented by the triplet \((\mathcal{V}, \mathcal{B}(\mathcal{V}), \Theta(\mathcal{B}(\mathcal{V})))\). Here, \(\mathcal{V}\) denotes the vine tree structure, \(\mathcal{B}(\mathcal{V})\) represents the collection of bivariate copula families assigned to the edges of \(\mathcal{V}\), and \(\Theta(\mathcal{B}(\mathcal{V}))\) denotes the corresponding set of copula parameters. In this paper, we employ the algorithm of \cite{dissmann2013selecting}, which is among the most widely used methods for selecting R-vine copula models.

To estimate \((\widehat{\mathcal{V}}, \widehat{\mathcal{B}}, \widehat{\Theta})\), we introduce  Algorithm~\ref{alg:discmann1} in Appendix~\ref{app:Diss}, which utilizes a weight function \(\tau\) based on the absolute values of empirical Kendall’s \(\tau\) statistics, as proposed by \cite{dissmann2013selecting}. The weight function \(\tau\) measures the strength of dependence between pairs of variables, and the algorithm seeks to maximize the total dependence among the selected edges in each tree. This is accomplished by constructing a maximum spanning tree (MST) at each level of the vine, where the MST is a connected acyclic graph whose total edge weight is maximal.

\begin{remark}\label{DissCD}
    Algorithm~\ref{alg:discmann1} in Appendix~\ref{app:Diss} may also be adapted to construct C-vine and D-vine copula models instead of a general R-vine. For C-vines, the root node at each tree level is selected as the node with the largest total edge weight, following the approach of \cite{czado2012maximum}. For D-vines, once the ordering of variables in the first tree is specified, the structures of all subsequent trees are uniquely determined. Since D-vine trees correspond to \emph{Hamiltonian paths}---paths that visit each node exactly once without repetition---the selection of an optimal ordering reduces to finding a maximum-weight Hamiltonian path, which is a variant of the traveling salesman problem; see \cite{brechmann2010truncated}.
\end{remark}

%%%%%%%%%%%%%%%%%%%%%%%%%%%%%%%%
\section{SL-vine copula}\label{sec3}
We introduce a novel vine structure, termed the \textit{stem-and-leaf} (SL) vine, within the framework of regular vines. We first define the SL-tree as a sequence of trees on a set of \(d\) elements and show that the resulting SL-vine satisfies the properties of an R-vine tree sequence. We then derive the associated SL-vine copula distribution. Unlike R-vines, which generally allow multiple valid configurations for higher-order trees, the SL-vine becomes uniquely determined after the initial step. To support practical implementation, we propose a systematic procedure for selecting the initial tree \(T_1\) and provide an algorithm for constructing the subsequent trees \(T_2,\dots,T_{d-1}\).
\subsection{SL-vine}\label{monggon:SL}
\begin{definition}\label{SLdef}
    In the \(k^{\text{th}}\) tree of the sequence, nodes forming a continuous path are termed \textit{stems} and collected in the set \(S_k\), while nodes attached to this path are referred to as \textit{leaves} and grouped in \(L_k\). For \(2m\) random variables, where \(m \in \mathbb{N}\), we define a sequence of \(2m - 1\) trees constituting the SL-vine structure as follows:
    \begin{align*}
         T_{1} &\text{ with  } N_1= S_1 \cup L_1 = \left\{1, 2,\dots, m\right\}\cup \left\{2m,\dots, m+1\right\} \text{ and }  \\ &E_1=\{(1,2), (2,3),\dots, (m,m+1)\}\cup\{ (2,2m), (3,2m-1),\dots,(m,m+2)\}.\\
        T_{2} &\text{ with } N_2=S_2 \cup L_2=\{(1,2),\dots, (m,m+1)\}\cup\{ (2,2m),\dots,(m,m+2)\} \text{ and }\\ &E_2=\{((1,2),(2,3)), ((2,3),(3,4)),\dots, ((m-1,m),(m,m+1))\}\\  &\quad\quad\cup\{ ((2,3),(2,2m)),\dots,((m,m+1),(m,m+2))\}\\
        &\quad= \left\{ (1, 3 \mid 2), (2, 4 \mid 3), \dots, (m-1, m+1 \mid m) \right\} \\ &\quad\quad\cup \left\{ (3, 2m \mid 2), (4, 2m-1 \mid 3), \dots, (m+1, m+2 \mid m) \right\}.\\
        \text{ For } &3\leq k\leq 2m-2. \text{ If } k=2l+1 \text{ for some } 1\leq l\leq m-2, \text{ then }
    \end{align*}
    \begin{align*}
        T_{k} &\text{ with }  N_k =S_k \cup L_k= \left\{ A_k^1, A_k^2, \dots, A_k^{m-l} \right\}\cup  \left\{ B_k^1, B_k^2, \dots, B_k^{m-l} \right\} \text{ and }\\ &E_k = \left\{ (A_k^1, A_k^2), (A_k^2, A_k^3), \dots, (A_k^{m-l-1}, A_k^{m-l}),(A_k^{m-l}, B_k^{m-l}) \right\}\\  &\quad\quad\cup \left\{ (A_k^2, B_k^1), (A_k^3, B_k^2), \dots, (A_k^{m-l}, B_k^{m-l-1}) \right\}.
        \\
        \text{ If } k&=2l+1 \text{ for some } 2\leq l\leq m-1, \text{ then }\\
        T_{k} &\text{ with }  N_k = S_k \cup L_k = \left\{ A_k^1, A_k^2, \dots, A_k^{m-l}, A_k^{m-l+1}\right\}\cup \left\{B_k^1, B_k^2, \dots, B_k^{m-l} \right\} \text{ and }\\
        &E_k = \left\{ (A_k^1, A_k^2), (A_k^2, A_k^3), \dots, (A_k^{m-l}, A_k^{m-l+1}) \right\}\\  &\quad\quad\cup \left\{ (A_k^2, B_k^1), (A_k^3, B_k^2), \dots, (A_k^{m-l}, B_k^{m-l-1}), (A_k^{m-l+1}, B_k^{m-l}) \right\}.\\
        &T_{2m-1} \text{ with }  N_{2m-1} = \left\{ A_{2m-1}^{1}, B_{2m-1}^{1} \right\} \text{ and } E_{2m-1} = \left\{ (A_{2m-1}^{1}, B_{2m-1}^{1}) \right\},
    \end{align*}
    where $ A_2^{i}:=(i,i+1) \mid \varnothing := (i,i+1)
 $ and $ B_2^{i}:= (i+1,2m-i+1) \mid \varnothing := (i+1,2m-i+1),$ and $ A_k^{i}:=(i,k-1+i) \mid (i+1:k-2+i)$ and
    $ B_k^{i}:=(k-1+i,2m-i+1) \mid (i+1:k-2+i)$ for $k\geq3$, $(x:y)$ denotes  $(x,x+1,\dots,y)$ and $(x,y):=\{x,y\}$.
\end{definition}
\begin{remark}
    Although the SL-vine has been defined for an even number of variables, it can be naturally extended to the case of \(2m - 1\) variables. In this setting, the SL-vine is represented by the sequence \((T_2, T_3, \dots, T_{2m-1})\), where \((T_1, T_2, \dots, T_{2m-1})\) denotes the corresponding R-vine tree sequence defined previously.
\end{remark}

    Figure \ref{monggon:fig3.1} illustrates the structure of the proposed SL-vine tree sequence. Moreover, this construction satisfies the conditions of an R-vine tree sequence as given in Definition \ref{R-vine}, thereby establishing its validity as a regular vine. A formal proof is provided below.

\begin{figure}[htbp]
\centering
\begin{tikzpicture}[every node/.style={thick, inner sep=1pt, font=\footnotesize}, node distance=1.1cm]

    % T1 layer
    \node (T1) at (-2,0) {$T_1$};
    
    \node[ellipse, draw] (2) [right=3cm of T1] {2};
    \node[ellipse, draw] (1) [left=of 2] {1};
    \node[ellipse, draw] (3) [right=of 2] {3};
    \node[ellipse, draw] (dots1) [right=of 3, draw=none] {$\dots$};
    \node[ellipse, draw] (m) [right=of dots1] {$m$};
    \node[ellipse, draw] (m+1) [right=of m] {$m+1$};
    \node[ellipse, draw] (m+2) [above left =of m] {$m+2$};
    \node[ellipse, draw] (2m-1) [above left =of 3] {$2m-1$};
    \node[ellipse, draw] (2m) [above left =of 2] {$2m$};

    \draw (1) -- (2);
    \draw (2) -- (3);
    \draw (3) -- (dots1);
    \draw (dots1) --  (m) ;
    \draw (m) --  (m+1) ;
    \draw (m+2) -- (m);
    \draw (3) --  (2m-1);
    \draw (2) --  (2m);
    
    % T2 layer
    \node (T2) at (-2,-1.75) {$T_2$};
    \node[ellipse, draw] (23) [right=3.5 cm of T2] {2,3};
    \node[ellipse, draw] (12) [left=of 23] {1,2};
    \node[ellipse, draw] (34) [right=of 23] {3,4};
    \node[ellipse, draw] (dots2) [right=of 34, draw=none] {$\dots$};
    \node[ellipse, draw] (m m+1) [right=of dots2] {$m,m+1$};
    \node[ellipse, draw] (m m+2) [above left =of m m+1] {$m,m+2$};
    \node[ellipse, draw] (3 2m-1) [above left =of 34 ] {$3,2m-1$};
    \node[ellipse, draw] (2 2m) [above left =of 23 ] {$2,2m$};
    
    \draw (12) -- (23);
    \draw (23) -- (34);
    \draw (34) -- (dots2);
    \draw (dots2) --  (m m+1) ;
    \draw (m m+1) --  (m m+2);
    \draw (3 2m-1) --  (34);
    \draw (2 2m) --  (23);

    % T2.5 layer
    \node (vdots) at (-2,-2.25) {$ \vdots $};
    \node (vdots2) at (4,-2.25) {$ \vdots $};
    
    % T2l+1 layer
    \node (T2l+1) at (-2,-4.25) {$T_{2l+1}$};
    
    \node[ellipse, draw] (Ak2) [right=3.25 cm of T2l+1] {$A_{2l+1}^{2}$};
    \node[ellipse, draw] (Ak1) [left=of Ak2] {$A_{2l+1}^{1}$};
    \node[ellipse, draw] (dots3) [right=of Ak2, draw=none] {$\dots$};
    \node[ellipse, draw] (Akm-l) [right=of dots3] {$A_{2l+1}^{m-l}$};
    \node[ellipse, draw] (Bkm-l) [right=of Akm-l] {$B_{2l+1}^{m-l}$};
    \node[ellipse, draw] (Bkm-l-1) [above left =of Akm-l] {$B_{2l+1}^{m-l-1}$};
    \node[ellipse, draw] (Bk1) [above left =of Ak2] {$B_{2l+1}^{1}$};

    \draw (Ak1) -- (Ak2);
    \draw (Ak2) -- (dots3);
    \draw (dots3) --  (Akm-l) ;
    \draw (Akm-l) --  (Bkm-l) ;
    \draw (Bkm-l-1) -- (Akm-l);
    \draw (Bk1) --  (Ak2);

    % T2l layer
    \node (T2l) at (-2,-6.2) {$T_{2l}$};
    
    \node[ellipse, draw] (Ak2) [right=4.25 cm of T2l] {$A_{2l}^{2}$};
    \node[ellipse, draw] (Ak1) [left=of Ak2] {$A_{2l}^{1}$};
    \node[ellipse, draw] (dots3) [right=of Ak2, draw=none] {$\dots$};
    \node[ellipse, draw] (Akm-l+1) [right=of dots3] {$A_{2l}^{m-l+1}$};
    \node[ellipse, draw] (Bkm-l) [above left =of Akm-l+1] {$B_{2l}^{m-l}$};
    \node[ellipse, draw] (Bk1) [above left =of Ak2] {$B_{2l}^{1}$};

    \draw (Ak1) -- (Ak2);
    \draw (Ak2) -- (dots3);
    \draw (dots3) --  (Akm-l+1) ;
    \draw (Bkm-l) -- (Akm-l+1);
    \draw (Bk1) --  (Ak2);

    % T2.5 layer
    \node (vdots) at (-2,-6.8) {$ \vdots $};
    \node (vdots2) at (4,-6.8) {$ \vdots $};
    
    % T2m-1 layer
    \node (T2m-1) at (-2,-7.4) {$T_{2m-1}$};
    \node[ellipse, draw] (1) [right=3.5 cm of T2m-1] {$A_{2m-1}^{1}$};
    \node[ellipse, draw] (2) [right=of 1] {$B_{2m-1}^{1}$};
    \draw (1) -- (2);
\end{tikzpicture}
\caption{Vine tree structures of an SL-vine} \label{monggon:fig3.1}
\end{figure}
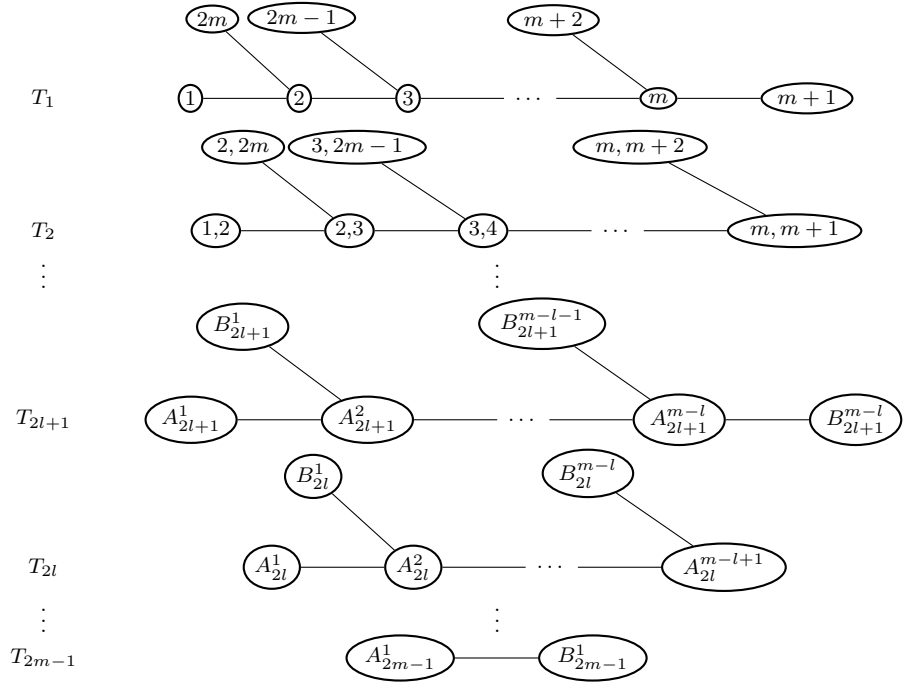
\begin{proposition}\label{monggon:prop1}
    The SL-vine tree sequence defined in Definition \ref{SLdef} satisfies conditions (i)–(iii) of the R-vine tree sequence given in Definition \ref{R-vine}.
\end{proposition}
\begin{proof}
     We begin by showing that the SL-vine tree sequence \(\{T_1, \dots, T_{2m-1}\}\) satisfies conditions (i)–(iii) as stated in Definition \ref{R-vine}. Condition (i) is immediately satisfied by the definition of \(T_1\). For condition (ii), which requires \(N_k = E_{k-1}\) for all \(k \geq 2\), the case \(k = 2\) directly follows from Definition \ref{SLdef}. To verify the condition for \(k \geq 3\), it suffices to demonstrate that \(A_k^i = (A_{k-1}^i, A_{k-1}^{i+1})\) and \(B_k^i = (A_{k-1}^{i+1}, B_{k-1}^i)\). For the former case, since
\begin{align*}
     A_{k-1}^{i}&= (i,k-2+i) \mid (i+1, i+2,\dots, k-3+i)\text{ and }\\
     A_{k-1}^{i+1}&=(i+1,k-1+i) \mid (i+2, i+3, \dots,k-2+i),
\end{align*}
the conditioning set 
\begin{align*}
    D_{(A_{k-1}^i, A_{k-1}^{i+1})}&=U_{A_{k-1}^i} \cap U_{A_{k-1}^{i+1}}\\
    &= \{{\color{blue}i},{\color{red}k-2+i, i+1, i+2,\dots,k-3+i}\} \\
    &\quad\cap \{{\color{red}i+1},{\color{blue}k-1+i}, {\color{red}i+2, i+3, \dots,k-2+i}\}\\
    &=\{{\color{red}i+1, i+2,\dots, k-2+i}\}=(i+1:k-2+i) \text{ and }
\end{align*}
the conditioned set
\begin{align*}
    {\Big(U_{A_{k-1}^i}\backslash D_{(A_{k-1}^i, A_{k-1}^{i+1})}, U_{A_{k-1}^{i+1}}\backslash D_{(A_{k-1}^i, A_{k-1}^{i+1})}\Big)}=({\color{blue}i},{\color{blue} k-1+i}). 
\end{align*}
Thus, $ (A_{k-1}^i, A_{k-1}^{i+1})=(i,k-1+i) \mid (i+1:k-2+i)=A_k^i $. For the latter case, since
\begin{align*}
     A_{k-1}^{i+1}&=(i+1,k-1+i) \mid (i+2,i+3, \dots, k-2+i) \text{ and }\\
     B_{k-1}^{i}&=(k-2+i,2m-i+1) \mid (i+1, i+2,\dots, k-3+i),
\end{align*}
the conditioning set
\begin{align*}
    D_{(A_{k-1}^{i+1}, B_{k-1}^i)}&=U_{A_{k-1}^{i+1}} \cap U_{B_{k-1}^i}\\
    &= \{{\color{red}i+1},{\color{blue}k-1+i}, {\color{red}i+2, i+3, \dots,k-2+i}\}\\
    &\quad\cap \{{\color{red}k-2+i},{\color{blue}2m-i+1},{\color{red} i+1, i+2,\dots, k-3+i}\}\\
    &=\{{\color{red}i+1, i+2,\dots, k-2+i}\}=(i+1:k-2+i) \text{ and }
\end{align*}
the conditioned set
\begin{align*}
    {\Big(U_{A_{k-1}^{i+1}}\backslash D_{(A_{k-1}^{i+1}, B_{k-1}^i)}, U_{B_{k-1}^i}\backslash D_{(A_{k-1}^{i+1}, B_{k-1}^i)}\Big)}=({\color{blue}k-1+i},{\color{blue}2m-i+1}). 
\end{align*}
 Therefore, $(A_{k-1}^{i+1}, B_{k-1}^i)=(k-i+1,2m-i+1)\mid (i+1:k-2+i)= B_k^i $. Consequently, according to Definition \ref{SLdef}, we have \( N_k = E_{k-1} \) for every \( k \) in the SL-vine tree sequence. The last condition to be verified is the proximity condition, which requires that for all \( k \geq 2 \), two nodes in \( T_k \) may be connected only if their corresponding edges in the previous tree \( T_{k-1} \) share a common node. For the case \( k = 2 \), this is straightforward. To establish the result for \( k \geq 3 \), we consider two separate cases.  Figures~\ref{Case I} and~\ref{Case II} illustrate the procedures for Cases~1 and~2, respectively.
\\
\textbf{Case I} $k=2l+1$ for some  $1\leq l\leq m-2$. Let $(e_1, e_2) \in E_k$.
    \begin{quote}
        \textbf{Case 1.1} $e_1 = {\color{teal}A_k^i}$, $e_2 = {\color{teal}A_k^{i+1}}$ for some $i = 1, 2, 3, \dots, m-l-1$. \\
        In $T_{k-1}$, the corresponding edges of $e_1$ and $e_2$ are $(A_{k-1}^i, A_{k-1}^{i+1})$ and $(A_{k-1}^{i+1}, A_{k-1}^{i+2})$, respectively, which have the common node $A_{k-1}^{i+1}$.

        \textbf{Case 1.2} $e_1 = {\color{teal}A_k^{m-l}}$, $e_2 = {\color{red}B_k^{m-l}}$. \\
        In $T_{k-1}$, the corresponding edges of $e_1$ and $e_2$ are $(A_{k-1}^{m-l}, A_{k-1}^{m-l+1})$ and $(A_{k-1}^{m-l+1}, B_{k-1}^{m-l})$, respectively, which have the common node $A_{k-1}^{m-l+1}$.
        
        \textbf{Case 1.3} $e_1 = {\color{teal}A_k^i}$, $e_2 = {\color{violet}B_k^{i-1}}$ for some $i = 2, 3, \dots, m-l$. \\
        In $T_{k-1}$, the corresponding edges of $e_1$ and $e_2$ are  $(A_{k-1}^i, A_{k-1}^{i+1})$ and $(A_{k-1}^i, B_{k-1}^{i-1})$,  respectively, which have the common node $A_{k-1}^i$.
    \end{quote}

\begin{figure}[htbp]
\centering
\begin{tikzpicture}[every node/.style={thick, inner sep=1pt, font=\scriptsize}, node distance=1cm]
    % T2 layer
    \node (T2) at (-2,-4.25) {$T_{k-1}$};
    \node[ellipse, draw] (23) [right=2.5 cm of T2] {$A_{k-1}^{2}$};
    \node[ellipse, draw] (12) [left=of 23] {$A_{k-1}^{1}$};
    \node[ellipse, draw] (34) [right=of 23] {$A_{k-1}^{3}$};
    \node[ellipse, draw] (dots2) [right= 0.5 cm of 34, draw=none] {$\dots$};
    \node[ellipse, draw] (m m+1) [right= 0.5 cm of dots2] {$A_{k-1}^{m-l}$};
    \node[ellipse, draw] (m m+2) [above left =of m m+1] {$B_{k-1}^{m-l-1}$};
    \node[ellipse, draw] (3 2m-1) [above left =of 34 ] {$B_{k-1}^{2}$};
    \node[ellipse, draw] (2 2m) [above left =of 23 ] {$B_{k-1}^{1}$};
    \node[ellipse, draw] (x) [right= 1.2 cm of m m+1] {$A_{k-1}^{m-l+1}$};
    \node[ellipse, draw] (y) [above left =of x] {$B_{k-1}^{m-l}$};

    \draw [teal](12) --node[below]{$A_{k}^{1}$} (23);
    \draw [teal](23) --node[below]{$A_{k}^{2}$} (34);
    \draw [teal](34) -- (dots2);
    \draw [teal](dots2) --  (m m+1) ;
    \draw [violet](m m+1) --node[above right]{$B_{k}^{m-l-1}$}  (m m+2);
    \draw [violet](3 2m-1) --node[above right]{$B_{k}^{2}$}  (34);
    \draw [violet](2 2m) --node[above right]{$B_{k}^{1}$}  (23);
    \draw [teal](m m+1) --node[below]{$A_{k}^{m-l}$}  (x);
    \draw [red](x) --node[above right]{$B_{k}^{m-l}$}  (y);

     % T2 layer
    \node (T2) at (-2,-6.4) {$T_{k}$};
    \node[ellipse, draw, teal] (23) [right=3 cm of T2] {$A_{k}^{2}$};
    \node[ellipse, draw, teal] (12) [left=of 23] {$A_{k}^{1}$};
    \node[ellipse, draw, teal] (34) [right=of 23] {$A_{k}^{3}$};
    \node[ellipse, draw] (dots2) [right= 0.5 cm of 34, draw=none] {$\dots$};
    \node[ellipse, draw, teal] (m m+1) [right= 0.5 cm of dots2] {$A_{k}^{m-l}$};
    \node[ellipse, draw, violet] (m m+2) [above left =of m m+1] {$B_{k}^{m-l-1}$};
    \node[ellipse, draw, violet] (3 2m-1) [above left =of 34 ] {$B_{k}^{2}$};
    \node[ellipse, draw, violet] (2 2m) [above left =of 23 ] {$B_{k}^{1}$};
    \node[ellipse, draw, red] (x) [right= 1.2 cm of m m+1] {$B_{k}^{m-l}$};

    \draw (12) -- (23);
    \draw (23) -- (34);
    \draw (34) -- (dots2);
    \draw (dots2) --  (m m+1) ;
    \draw (m m+1) --  (m m+2);
    \draw (3 2m-1) --  (34);
    \draw (2 2m) --  (23);
    \draw (m m+1) --  (x);
    
\end{tikzpicture}
\caption{Case I} \label{Case I}
\end{figure}
    
    \noindent \textbf{Case II} $k=2l$  for some $2\leq l\leq m-1$. Let $(e_1, e_2) \in E_k$.

    \begin{quote}
        \textbf{Case 2.1} $e_1 = {\color{teal}A_k^i}$, $e_2 = {\color{teal}A_k^{i+1}}$ for some $i = 1, 2, 3, \dots, m-l-1$. \\
        In $T_{k-1}$, the corresponding edges of $e_1$ and $e_2$ are $(A_{k-1}^i, A_{k-1}^{i+1})$  and $(A_{k-1}^{i+1}, A_{k-1}^{i+2})$, respectively, which have the common node $A_{k-1}^{i+1}$.

        \textbf{Case 2.2} $e_1 = {\color{teal}A_k^{m-l}}$, $e_2 = {\color{red}A_k^{m-l+1}}$. \\
        In $T_{k-1}$, the corresponding edges of $e_1$ and $e_2$ are $(A_{k-1}^{m-l}, A_{k-1}^{m-l+1})$ and $(A_{k-1}^{m-l+1}, B_{k-1}^{m-l+1})$, respectively, which have the common node $A_{k-1}^{m-l+1}$.

        \textbf{Case 2.3} $e_1 = {\color{teal}A_k^i}$, $e_2 = {\color{violet}B_k^{i-1}}$ for some $i = 2, 3, \dots, {\color{red}m-l+1}$. \\
        In $T_{k-1}$, the corresponding edges of $e_1$ and $e_2$ are $(A_{k-1}^i, A_{k-1}^{i+1})$ and  $(A_{k-1}^i, B_{k-1}^{i-1})$, respectively, which have the common node $A_{k-1}^i$.
\begin{figure}[htbp]
\centering
\begin{tikzpicture}[every node/.style={thick, inner sep=1pt, font=\scriptsize}, node distance=1cm]
    % T2 layer
    \node (T2) at (-2,-2.25) {$T_{k-1}$};
    \node[ellipse, draw] (23) [right=2.5 cm of T2] {$A_{k-1}^{2}$};
    \node[ellipse, draw] (12) [left=of 23] {$A_{k-1}^{1}$};
    \node[ellipse, draw] (34) [right=of 23] {$A_{k-1}^{3}$};
    \node[ellipse, draw] (dots2) [right= 0.25 cm of 34, draw=none] {$\dots$};
    \node[ellipse, draw] (m m+1) [right=0.25 cm of dots2] {$A_{k-1}^{m-l}$};
    \node[ellipse, draw] (m m+2) [above left =of m m+1] {$B_{k-1}^{m-l-1}$};
    \node[ellipse, draw] (3 2m-1) [above left =of 34 ] {$B_{k-1}^{2}$};
    \node[ellipse, draw] (2 2m) [above left =of 23 ] {$B_{k-1}^{1}$};
    \node[ellipse, draw] (x) [right=of m m+1] {$A_{k-1}^{m-l+1}$};
    \node[ellipse, draw] (y) [above left =of x] {$B_{k-1}^{m-l}$};
    \node[ellipse, draw] (z) [right=of x] {$B_{k-1}^{m-l+1}$};
    
    \draw [teal](12) --node[below]{$A_{k}^{1}$} (23);
    \draw [teal](23) --node[below]{$A_{k}^{2}$} (34);
    \draw [teal](34) -- (dots2);
    \draw [teal](dots2) --  (m m+1) ;
    \draw [violet](m m+1) --node[above right]{$B_{k}^{m-l-1}$}  (m m+2);
    \draw [violet](3 2m-1) --node[above right]{$B_{k}^{2}$}  (34);
    \draw [violet](2 2m) --node[above right]{$B_{k}^{1}$}  (23);
    \draw [teal](m m+1) --node[below]{$A_{k}^{m-l}$}  (x);
    \draw [violet](x) --node[above right]{$B_{k}^{m-l}$}  (y);
    \draw [red](x) --node[below]{$A_{k}^{m-l+1}$}  (z);

     % T2 layer
    \node (T2) at (-2,-4.4) {$T_{k}$};
    \node[ellipse, draw, teal] (23) [right=3.5 cm of T2] {$A_{k}^{2}$};
    \node[ellipse, draw, teal] (12) [left=of 23] {$A_{k}^{1}$};
    \node[ellipse, draw, teal] (34) [right=of 23] {$A_{k}^{3}$};
    \node[ellipse, draw] (dots2) [right= 0.25 cm of 34, draw=none] {$\dots$};
    \node[ellipse, draw, teal] (m m+1) [right=0.25 cm of dots2] {$A_{k}^{m-l}$};
    \node[ellipse, draw, violet] (m m+2) [above left =of m m+1] {$B_{k}^{m-l-1}$};
    \node[ellipse, draw, violet] (3 2m-1) [above left =of 34 ] {$B_{k}^{2}$};
    \node[ellipse, draw, violet] (2 2m) [above left =of 23 ] {$B_{k}^{1}$};
    \node[ellipse, draw, red] (x) [right=of m m+1] {$A_{k}^{m-l+1}$};
    \node[ellipse, draw, violet] (y) [above left =of x] {$B_{k}^{m-l}$};

    \draw (12) -- (23);
    \draw (23) -- (34);
    \draw (34) -- (dots2);
    \draw (dots2) --  (m m+1) ;
    \draw (m m+1) --  (m m+2);
    \draw (3 2m-1) --  (34);
    \draw (2 2m) --  (23);
    \draw (m m+1) --  (x);
    \draw (x) --  (y);

\end{tikzpicture}
\caption{Case II}\label{Case II}
\end{figure}
    \end{quote}
\end{proof}

 Hence, the SL-vine tree sequence satisfies the R-vine conditions in Definition \ref{R-vine}. A regular vine copula constructed from an SL-vine tree sequence \(\mathcal{V}\), together with a collection of bivariate copulas \(\mathcal{B}(\mathcal{V})\) and their associated parameters \(\boldsymbol{\theta}(\mathcal{B}(\mathcal{V}))\), is referred to as an SL-vine copula, denoted by \(\mathcal{SL}\). The following theorem provides the construction of the SL-vine distribution for a \(2m\)-dimensional random vector \(\mathbf{X} = (X_1, \dots, X_{2m})\).

\begin{theorem}\label{monggon:thm1}
    The joint density \(f(x_1, \dots, x_{2m})\) induced by the SL-vine copula is given by
  \begin{align}\label{monggon:eq5}
      f(x_1, \dots, x_{2m}) &= 
    \Bigg[\prod_{\ell=1}^{m-1}  \prod_{j=1}^{m-\ell-1} c_{A_{2\ell+2}^j}\Bigg] \times
    \Bigg[\prod_{\ell=1}^{m-1}  \prod_{j=1}^{m-\ell-1} c_{B_{2\ell+2}^j} 
    \Bigg]\times\Bigg[\prod_{\ell=1}^{m-1} c_{A_{2\ell+1}^{m-\ell}, B_{2\ell+1}^{m-\ell}}\Bigg]\notag\\ 
    &\quad\times 
    \Bigg[ \prod_{\ell=2}^{m-1}\prod_{j=1}^{m-\ell} c_{A_{2\ell+1}^j}\Bigg] \times
    \Bigg[ \prod_{\ell=2}^{m-1}\prod_{j=1}^{m-\ell} c_{ B_{2\ell+1}^j}\Bigg]\times\Bigg[\prod_{i=1}^{m-1} c_{(i,i+2)|(i+1)}\Bigg]\notag\\
    &\quad\times\Bigg[\prod_{i=m}^{2m-2} c_{(2m-i+1,i+2)|(2m-i)}\Bigg]\times\Bigg[\prod_{i=1}^{m} c_{i,i+1}\Bigg]\notag\\
    &\quad\times\Bigg[\prod_{i=m+1}^{2m-1} c_{(2m-i+1,i+1)}\Bigg]\times\Bigg[\prod_{i=1}^{2m} f_i(x_i)\Bigg],
  \end{align}
  where $ A_k^{i}$ and 
    $ B_k^{i}$ are defined in Definition \ref{SLdef}.
\end{theorem}
\begin{proof}
        The joint density \(f(x_1, \dots, x_{2m})\) induced by the SL-vine tree sequence in Definition~\ref{SLdef} is given in equation~(\ref{monggon:eq4}) as
    \begin{align*}
      f(x_1, \dots, x_{2m}) &= \overbrace{\prod_{\ell=1}^{m-1} \Bigg[
    \Bigg( \prod_{j=1}^{m-\ell-1} c_{A_{2\ell+1}^j, A_{2\ell+1}^{j+1}}\Bigg) 
    \Bigg( \prod_{j=1}^{m-\ell-1} c_{A_{2\ell+1}^{j+1}, B_{2\ell+1}^j} \Bigg) 
    c_{A_{2\ell+1}^{m-\ell}, B_{2\ell+1}^{m-\ell}}\Bigg]}^{T_{2l+1}}\\ 
    &\quad\times\overbrace{\prod_{\ell=2}^{m-1} \Bigg[
    \Bigg( \prod_{j=1}^{m-\ell} c_{A_{2\ell}^j, A_{2\ell}^{j+1}}\Bigg) 
    \Bigg( \prod_{j=1}^{m-\ell} c_{A_{2\ell}^{j+1}, B_{2\ell}^j}\Bigg)\Bigg]}^{T_{2l}}\times\Bigg[\prod_{i=1}^{2m} f_i(x_i)\Bigg]\\
    &\quad\times\overbrace{\Bigg[\Bigg(\prod_{i=1}^{m-1} c_{(i,i+2)|(i+1)}\Bigg)\Bigg(\prod_{i=m}^{2m-2} c_{(2m-i+1,i+2)|(2m-i)}\Bigg)\Bigg]}^{T_{2}}\\ &\quad\times\overbrace{\Bigg[\Bigg(\prod_{i=1}^{m}c_{i,i+1}\Bigg)\Bigg(\prod_{i=m+1}^{2m-1} c_{(2m-i+1,i+1)}\Bigg)\Bigg]}^{T_1}.
  \end{align*}
      Since $A_k^i = (A_{k-1}^i, A_{k-1}^{i+1})$ and $B_k^i = (A_{k-1}^{i+1}, B_{k-1}^i)$, equation (\ref{monggon:eq5}) holds.
\end{proof}

By definition, once the first tree \(T_1\) is specified, each subsequent tree \(T_k\), for \(k = 2, \dots, 2m-1\), is uniquely determined in the SL-vine. Nevertheless, additional flexibility can be introduced in constructing \(T_2, \dots, T_{2m-1}\) while preserving the stem-and-leaf structure, by allowing alternative edge selections. For illustration, Figure \ref{monggon:fig3.2} displays all candidate edges for constructing \(T_2\) in the R-vine case with \(m = 3\), given that \(T_1\) follows the SL-vine structure. To ensure that the resulting graph remains connected and acyclic, the selection of edges must satisfy these constraints. All admissible connected trees that retain a structure similar to \(T_2\) in Figure \ref{monggon:fig3.1} are shown in Figure \ref{monggon:fig3.3}. The subsequent trees \(T_3, T_4,\) and \(T_5\) follow the same structural pattern.

\begin{figure}[htbp]
\centering
\begin{tikzpicture}[every node/.style={thick, inner sep=1pt, font=\small}, node distance=1.25 cm]

    % T2 layer
    \node (T2) at (-1,-1.75) {$T_1 \to T_2$};
    \node[ellipse, draw] (23) [right=3 cm of T2] {2,3};
    \node[ellipse, draw] (12) [left=of 23] {1,2};
    \node[ellipse, draw] (34) [right=of 23] {3,4};
    \node[ellipse, draw] (35) [above left =of 34 ] {$3,5$};
    \node[ellipse, draw] (26) [above left =of 23 ] {$2,6$};
    
    \draw (12) -- (23);
    \draw (12) -- (26);
    \draw (23) -- (34);
    \draw (23) -- (35);
    \draw (35) --  (34);
    \draw (26) --  (23);
    
\end{tikzpicture}
\caption{All possible edges in constructing \( T_2 \) for the R-vine case with \( m = 3 \)}
    \label{monggon:fig3.2}
\end{figure}
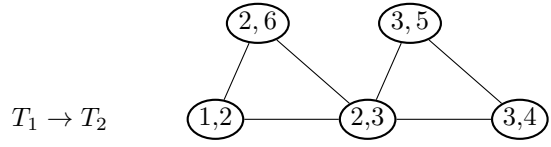

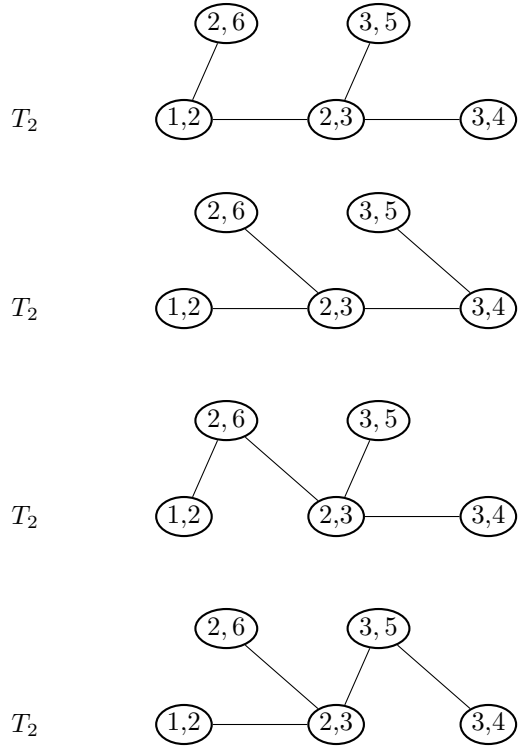
\begin{figure}[htbp]
\centering
\begin{tikzpicture}[every node/.style={thick, inner sep=1pt, font=\small}, node distance=1.25 cm]

    % T2 layer
    \node (T2) at (-2,-1.75) {$T_2$};
    \node[ellipse, draw] (23) [right=3.5 cm of T2] {2,3};
    \node[ellipse, draw] (12) [left=of 23] {1,2};
    \node[ellipse, draw] (34) [right=of 23] {3,4};
    \node[ellipse, draw] (35) [above left =of 34 ] {$3,5$};
    \node[ellipse, draw] (26) [above left =of 23 ] {$2,6$};
    
    \draw (12) -- (23);
    \draw (12) -- (26);
    \draw (23) -- (34);
    \draw (23) -- (35);

    \node (T2) at (-2,-4.25) {$T_2$};
    \node[ellipse, draw] (23) [right=3.5 cm of T2] {2,3};
    \node[ellipse, draw] (12) [left=of 23] {1,2};
    \node[ellipse, draw] (34) [right=of 23] {3,4};
    \node[ellipse, draw] (35) [above left =of 34 ] {$3,5$};
    \node[ellipse, draw] (26) [above left =of 23 ] {$2,6$};
    
    \draw (12) -- (23);
    \draw (23) -- (34);
    \draw (35) --  (34);
    \draw (26) --  (23);

    \node (T2) at (-2,-7) {$T_2$};
    \node[ellipse, draw] (23) [right=3.5 cm of T2] {2,3};
    \node[ellipse, draw] (12) [left=of 23] {1,2};
    \node[ellipse, draw] (34) [right=of 23] {3,4};
    \node[ellipse, draw] (35) [above left =of 34 ] {$3,5$};
    \node[ellipse, draw] (26) [above left =of 23 ] {$2,6$};
    
    \draw (12) -- (26);
    \draw (23) -- (34);
    \draw (23) -- (35);
    \draw (26) --  (23);

    \node (T2) at (-2,-9.75) {$T_2$};
    \node[ellipse, draw] (23) [right=3.5 cm of T2] {2,3};
    \node[ellipse, draw] (12) [left=of 23] {1,2};
    \node[ellipse, draw] (34) [right=of 23] {3,4};
    \node[ellipse, draw] (35) [above left =of 34 ] {$3,5$};
    \node[ellipse, draw] (26) [above left =of 23 ] {$2,6$};
    
    \draw (12) -- (23);
    \draw (23) -- (35);
    \draw (35) --  (34);
    \draw (26) --  (23);
\end{tikzpicture}
\caption{All possible connected trees corresponding to \( T_2 \)}
    \label{monggon:fig3.3}
\end{figure}
\subsection{Algorithm for SL-vine}\label{AlgSL}
The construction of each vine structure is governed by Steps 2 and 12 in Algorithm~\ref{alg:discmann1}. Motivated by this observation, we develop a procedure for determining \(T_1\) and an accompanying algorithm for generating \(T_d\), \(d = 2, \dots, 2m-1\), in accordance with Definition~\ref{SLdef}, thereby facilitating the application of the SL-vine to real-world data.

The construction of \(T_1\) follows the approach used for D-vines, as discussed in Remark~\ref{DissCD}. Specifically, the first \(m\) nodes are arranged into a path forming the stem \(S_1\). The remaining \(m\) nodes, corresponding to the leaves \(L_1\), are then assigned to nodes in \(S_1\) through a pairwise matching scheme, excluding one terminal node of the stem. To determine this assignment, we employ the Hungarian algorithm of \cite{kuhn1955hungarian}, which solves the underlying one-to-one matching problem by minimizing the overall assignment cost.

To generate the sequence of trees \(T_d\), \(d = 2, \dots, 2m-1\), we first develop Algorithm~\ref{alg:getproblem} in Appendix~\ref{app:AlgSL}, called GetProblem. Starting from the graph \(T_1 = S_1 \cup L_1\), which satisfies the proximity condition, the algorithm identifies a collection of structural configurations, referred to as \textit{Problems}. An illustration of the input graph is provided in Figure~\ref{monggon:fig3.2}. These configurations serve as the foundation for Algorithm~\ref{alg:slvineselection} in Appendix~\ref{app:AlgSL}, denoted by SL-vine Selection, which iteratively constructs the remaining trees in the SL-vine sequence. Within these algorithms, a \textit{point} represents a graph vertex indexed by a positive integer. An \textit{edge} is defined as a \(1 \times 2\) matrix specifying two connected points, while a \textit{triangle} denotes a \(1 \times 3\) matrix containing three pairwise adjacent points.

%%%%%%%%%%%%%%%%%%%%%%%%%%%%%%%%%
\section{\texorpdfstring{$\Delta$}{Delta}-vine copula}\label{sec4}

As noted in Section~\ref{monggon:SL}, the structure of an SL-vine becomes uniquely determined once \(T_1\) is specified, although limited flexibility can still be incorporated. This motivates the question of whether a more flexible vine structure can be developed without imposing such restrictions on \(T_1\). Figure~\ref{monggon:fig3.3} shows that, across all admissible tree configurations, each node is connected to at most three edges. Motivated by this observation, we introduce a new class of vine structures, referred to as the \(\Delta\)-vine, which generalizes the stem-and-leaf (SL) vine through a degree constraint of at most three. We begin by defining the associated \(\Delta\)-tree sequences on \(d\) elements and then establish that the resulting \(\Delta\)-vine satisfies the R-vine tree sequence conditions. We further develop a procedure for selecting the initial tree \(T_1\) and propose an algorithm for constructing the subsequent trees \(T_k\), \(k \geq 2\). The proposed \(\Delta\)-vine is designed to balance structural flexibility and model simplicity in multivariate dependence modeling. By restricting each node to degree at most three, it avoids the rigidity of the SL-vine while remaining compatible with the R-vine framework, thereby preserving the applicability of existing theoretical results and computational tools. In addition, its explicit cycle-elimination mechanism facilitates algorithmic implementation and supports automated, interpretable vine construction.
\subsection{\texorpdfstring{$\Delta$}{Delta}-vine}
As noted in Remark~\ref{rem: R-vine}, the graph induced by the connected edges of \(T_k\) may contain cycles, requiring certain edges to be removed in order to obtain the tree \(T_{k+1}\). To address this issue, we first introduce the concept of a line graph, which characterizes all admissible edges, as defined below.

\begin{definition}
For any tree $T = (V, E)$, let us denote by $L(T)$ the \textbf{line graph} of $T$, of which the vertex set is exactly the set $E$ and in which any two vertices $v, w \in E$ are adjacent whenever the corresponding edges in $T$ meet at a common vertex.
\end{definition}

\begin{example}
    Consider the vertex set \(\{1,2,3,4\}\). Figure~\ref{fig: T and L(T)} displays the tree \(T\) on the left together with its associated line graph \(L(T)\) on the right.
\end{example}
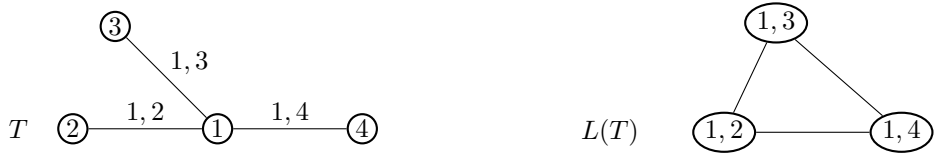
\begin{figure}[htbp]
    \centering
\begin{minipage}{0.48\textwidth}  % Left column (Trees)
    \centering
        \centering
        \begin{tikzpicture}[every node/.style={thick, inner sep=1pt, font=\small}, node distance=1.5 cm]

            % T2 layer
            \node (T2) at (-1,-1.75) {$T$};
            \node[circle, draw] (1) [right=2.25 cm of T2] {$1$};
            \node[circle, draw] (2) [left=of 1] {$2$};
            \node[circle, draw] (4) [right=of 1] {$4$};
            \node[circle, draw] (3) [above left=of 1] {$3$};

            \draw (2) -- node[above] {$1,2$} (1);
            \draw (1) -- node[above] {$1,4$} (4);
            \draw (3) -- node[above right] {$1,3$} (1);
        \end{tikzpicture}
\end{minipage}
\hfill
\begin{minipage}{0.48\textwidth}  % Right column (Pair-copulas)
    \centering
        \centering
        \begin{tikzpicture}[every node/.style={thick, inner sep=1pt, font=\small}, node distance=1.5 cm]

            % T2 layer
    \node (T2) at (-1,-1.75) {$L(T)$};
    \node[ellipse, draw] (14) [right=3 cm of T2] {$1,4$};
    \node[ellipse, draw] (12) [left=of 14] {$1,2$};
    \node[ellipse, draw] (13) [above left =of 14 ] {$1,3$};
    
    \draw (12) -- (14);
    \draw (12) -- (13);
    \draw (13) --  (14);
        \end{tikzpicture}
\end{minipage}
\caption{ $T$ and $L(T)$ graphs on $\{1,2,3,4\}$. }
    \label{fig: T and L(T)}
\end{figure}
\newpage
We now introduce the \(\Delta\)-vine tree sequence, in which every node of the \(k^{\text{th}}\) tree \(T_k\) has degree at most three. The construction starts from an initial tree satisfying this degree constraint. At each stage, a new graph is formed by treating the edges of the previous tree as vertices and connecting them according to the proximity condition. This procedure may produce 3-cycles, requiring the removal of selected edges to preserve acyclicity while maintaining connectivity. Consequently, the graph obtained at every step must remain a connected tree. The formal definition of the \(\Delta\)-vine is given below.

\begin{definition}\label{def: delta tree seq}
The sequence of trees $T_1, \dots, T_{d-1}$ is called a \emph{$\Delta$-vine} tree sequence if it satisfies:
\begin{enumerate}
    \item \(T_1=(V_1,E_1)\) is a tree with \(|V_1|=d\) vertices, and every vertex has degree at most \(3\)

    \item For each $n = 1, \dots, d-2$, there exists a sequence of connected graphs $G^n_0, G^n_1, \dots, G^n_m$, where $G^n_0 = L(T_n), G^n_m = T_{n+1}$ and $m, m-1, \dots, 1, 0$ are the numbers of 3-cycles in the graphs $G^n_0, G^n_1, \dots, G^n_{m-1}, G^n_m$, respectively. Note that $m$ is the numbers of 3-cycles in $L(T_n)$, and all $G^n_{k}$ are defined on the same vertex set. The sequence is constructed via the following process that removes one edge from each 3-cycles.

    \begin{enumerate}
        \item[2.1.] Let $\Delta$ be the set of all $3$-cycles in $G^n_0$, and set $D_0 = \emptyset$, $S_0 = \emptyset$.
        
        \item[2.2.] For each $k = 0, 1, 2, \dots, m-1$, suppose that $D_k$ is the set of $k$ $3$-cycles that already have one edge removed and $S_k$ is the set of $k$ deleted edges.
        \begin{enumerate}
            \item[2.2.1.] We shall choose a $3$-cycle $C_{k+1}$ in $\Delta-D_k$ and an edge $e_{k+1}$ in $C_{k+1}$ by considering the following two cases.
            \begin{enumerate}
                \item[i.] If there is no pair $(C, C') \in (\Delta - D_k) \times D_k$ such that $C$ and $C'$ share a common vertex, then choose any $C \in \Delta - D_k$, and select the edge in $C$ with the minimum weight. An example of a weight is $|\tau|$,

                \item[ii.] If there exists a pair $(C, C') \in (\Delta - D_k) \times D_k$ such that $C$ and $C'$ share a common vertex, namely $v$, then select an edge in $C$ using the following rules:
                \begin{itemize}
                    \item[a.] If $\deg_{G^n_k}(v) < 4$, then select the edge in $C$  with the minimum weight.

                    \item[b.] If $\deg_{G^n_k}(v) = 4$, then select the edge with the minimum weight in $C$ among those containing the vertex $v$.
                \end{itemize} 
            \end{enumerate}
            
        \item[2.2.2.] Set $D_{k+1} = D_k \cup \{C_{k+1}\}$ and $S_{k+1} = S_k \cup \{e_{k+1}\}.$
        \item[2.2.3.] Set $E(G^n_{k+1}) = E(L(T_n)) - S_{k+1}.$
        \end{enumerate}
    \end{enumerate}
\end{enumerate}
\end{definition}
\begin{remark}
    By condition~1 of Definition~\ref{def: delta tree seq}, the graph \(T_{n+1}\) must be constructed as a subgraph of \(L(T_n)\). In order for \(T_{n+1}\) to remain a tree, \(m\) edges must be eliminated from \(L(T_n)\). To formalize this iterative edge-removal process and obtain the characteristic structure of the \(\Delta\)-vine, we introduce the sequence \(\{G_i^n\}\). 
\end{remark}
\begin{remark}
    In the vine setting, the line graph \(L(G)\) cannot contain two distinct \(3\)-cycles sharing two common vertices, since this would imply the existence of a \(3\)-cycle in the original graph \(G\). Such a configuration contradicts the defining property of a vine, where \(G\) must be a tree. Figure~\ref{fig:GvsLG} illustrates this situation, with the graph \(G\) shown on the left and its corresponding line graph \(L(G)\) displayed on the right.
    \begin{figure}[htbp]
    \centering
\begin{minipage}{0.48\textwidth}  % Left column (Trees)
    \centering
        \centering
        \begin{tikzpicture}[every node/.style={thick, inner sep=1pt, font=\normalsize}, node distance=2 cm]

            % T2 layer
    \node[circle, draw] (2)  at (-1,-1.75) {$a$};
    \node[circle, draw] (1) [left=of 2] {$b$};
    \node[circle, draw] (4) [above  =of 1 ] {$d$};
    \node[circle, draw] (3) [right =of 4 ] {$c$};
    \node (G) at (-4.5,-1.75) {$G$};
    
    \draw (1) --node[below]{$e_1$} (2);
    \draw (1) --node[below left]{$e_2$} (4);
    \draw (3) --node[above]{$e_3$} (4);
    \draw (4) --node[above right ]{$e_4$} (2);
        \end{tikzpicture}
\end{minipage}
\hfill
\begin{minipage}{0.48\textwidth}  % Right column (Pair-copulas)
    \centering
        \centering
        \begin{tikzpicture}[every node/.style={thick, inner sep=1pt, font=\normalsize}, node distance=2 cm]

          % T2 layer
    \node[circle, draw] (2)  at (-1,-1.75) {$e_1$};
    \node[circle, draw] (1) [left=of 2] {$e_2$};
    \node[circle, draw] (4) [above  =of 1 ] {$e_3$};
    \node[circle, draw] (3) [right =of 4 ] {$e_4$};
    \node (LG) at (-4.5,-1.75) {$L(G)$};
    
    \draw (1) --(2);
    \draw (1) --(4);
    \draw (3) --(4);
    \draw (1) --(3);
    \draw (3) --(2);
    
        \end{tikzpicture}
\end{minipage}
\caption{ $G$ and $L(G)$ graphs on $\{a,b,c,d\}$.}
    \label{fig:GvsLG}
\end{figure}
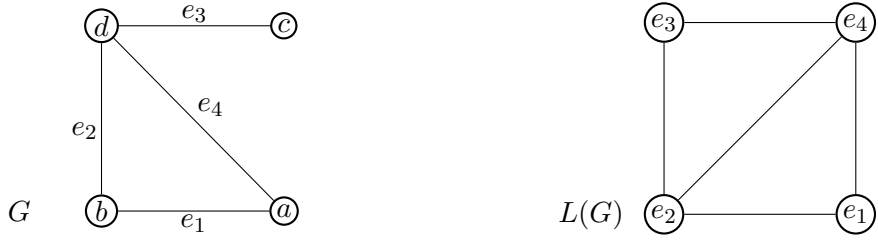
\end{remark}

\begin{remark}\label{rem 4.4}
    The \(\Delta\)-vine tree sequence defined in Definition~\ref{def:  delta tree seq} satisfies conditions (i)--(iii) of the R-vine tree sequence in Definition~\ref{R-vine}. Condition~(i) follows directly from part~(1) of Definition~\ref{def:  delta tree seq}, while condition~(ii) is guaranteed by the cycle-elimination procedure described in step~(2.2). Finally, condition~(iii), namely the proximity condition, holds because it is satisfied by the initial graph \(G_0^n\), and the removal of edges does not violate this property.
\end{remark}
\begin{example}\label{ex: v12}
    Consider the vertex set $\{1,\dots,12\}$ with $n=1$. Figures~\ref{DelT1}--\ref{DelG5} illustrate the construction of the tree $T_2$ according to Definition~\ref{def: delta tree seq}.
\end{example}

\begin{enumerate}
    \item[1.] Let $T_1 = (V_1, E_1)$ with the node set $V_1 = \{1, \dots, 12\}$ and the edge set $E_1 = \{(k,k+1):k=1,2,\dots,5\}\cup \{\{5,7\}, \{4,12\}, \{3,8\}, \{2,9\}, \{8,10\}, \{8,11\}\}.$
\end{enumerate}
\begin{figure}[htbp]
\centering
\begin{tikzpicture}[every node/.style={thick, inner sep=1pt, font=\normalsize}, node distance=1 cm]

    % T1 layer
    \node (T1) at (-1,-1.75) {$T_1$};
    \node[circle, draw] (2) [right=3.5 cm of T1] {$2$};
    \node[circle, draw] (1) [left=of 2] {$1$};
    \node[circle, draw] (3) [right=of 2] {$3$};
    \node[circle, draw] (4) [right=of 3] {$4$};
    \node[circle, draw] (5) [right=of 4] {$5$};
    \node[circle, draw] (6) [right=of 5] {$6$};
    \node[circle, draw] (7) [above left =of 5] {$7$};
    \node[circle, draw] (12) [above left =of 4] {$12$};
    \node[circle, draw] (8) [above left =of 3 ] {$8$};
    \node[circle, draw] (10) [above left =of 8 ] {$10$};
    \node[circle, draw] (11) [above left =of 12 ] {$11$};
    \node[circle, draw] (9) [above left =of 2 ] {$9$};
    
    \draw (1) -- node[above] {$1,2$} (2);
    \draw (2) -- node[above] {$2,3$} (3);
    \draw (3) -- node[above] {$3,4$} (4);
    \draw (4) -- node[above] {$4,5$} (5);
    \draw (5) -- node[above] {$5,6$} (6);
    \draw (5) -- node[above right] {$5,7$} (7);
    \draw (12) -- node[above right] {$4,12$} (4);
    \draw (8) -- node[above right] {$3,8$} (3);
    \draw (8) -- node[below left] {$8,10$} (10);
    \draw (8) -- node[below right] {$8,11$} (11);
    \draw (9) -- node[above right] {$2,9$} (2);
    
\end{tikzpicture}
\caption{The graph $T_1=(V_1,E_1)$.}
\label{DelT1}
\end{figure}
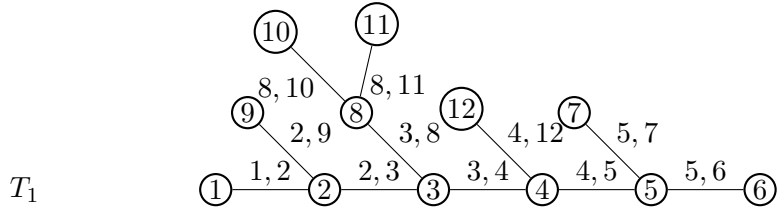

\begin{itemize}
    \item[2.1.] $\Delta=\{C_1, C_2, C_3, C_4, C_5\}$ and $D_0 = \emptyset$, $S_0 = \emptyset.$
\end{itemize}

\begin{figure}[htbp]
\centering
\begin{tikzpicture}[every node/.style={thick, inner sep=1pt, font=\normalsize}, node distance=1 cm]

    % T1 layer
    \node (T1) at (-1,-1.75) {$G^{1}_{0}=L(T_1)$};
    \node[ellipse, draw] (23) [right=3.5 cm of T1] {$2,3$};
    \node[ellipse, draw] (12) [left=of 23] {$1,2$};
    \node[ellipse, draw] (34) [right=of 23] {$3,4$};
    \node[ellipse, draw] (45) [right=of 34] {$4,5$};
    \node[ellipse, draw] (56) [right=of 45] {$5,6$};
    \node[ellipse, draw] (57) [above left =of 56] {$5,7$};
    \node[ellipse, draw] (412) [above left =of 45] {$4,12$};
    \node[ellipse, draw] (38) [above left =of 34 ] {$3,8$};
    \node[ellipse, draw] (810) [above left =of 38 ] {$8,10$};
    \node[ellipse, draw] (811) [above left =of 412 ] {$8,11$};
    \node[ellipse, draw] (29) [above left =of 23 ] {$2,9$};
    
    \draw (12) -- node[above ] {$C_1$} (23);
    \draw (12) -- (29);
    \draw (23) -- node[above ] {$C_2$} (34);
    \draw (34) -- node[above ] {$C_3$} (45);
    \draw (45) -- node[above ] {$C_4$} (56);
    \draw (45) -- (57);
    \draw (56) -- (57);
    \draw (34) -- (412);
    \draw (45) -- (412);
    \draw (23) -- (38);
    \draw (38) --  (34);
    \draw (29) --  (23);
    \draw (810) --  node[below ] {$C_5$} (811);
    \draw (38) --  (810);
    \draw (38) --  (811);
\end{tikzpicture}
\caption{The graph $G^{1}_{0}$.}
\label{DelG0}
\end{figure}
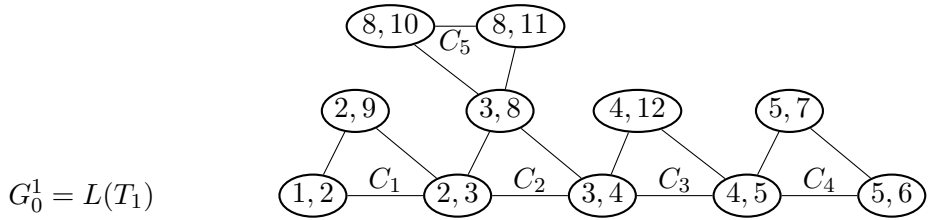

\begin{itemize}
    \item[2.2.] For $k=0$, $D_0=\emptyset$ and $S_0=\emptyset$. \\ 
    2.2.1(i). Choose $C_1 \in \Delta - D_0$ and select the edge in $C_1$ with the minimum weight, namely $e_1$ .\\
    2.2.2. Set $D_1 = D_0 \cup \{C_1\}= \{C_1\}$ and $S_{1}=S_{0} \cup \{ e_1\}= \{e_1\}.$\\
    2.2.3. Set $E(G^1_{1}) = E(L(T_1)) - S_{1}.$
\end{itemize}

\begin{figure}[htbp]
\centering
\begin{tikzpicture}[every node/.style={thick, inner sep=1pt, font=\normalsize}, node distance=1 cm]

    % T1 layer
    \node (T1) at (-1,-1.75) {$G^{1}_{1}$};
    \node[ellipse, draw] (23) [right=3.5 cm of T1] {$2,3$};
    \node[ellipse, draw] (12) [left=of 23] {$1,2$};
    \node[ellipse, draw] (34) [right=of 23] {$3,4$};
    \node[ellipse, draw] (45) [right=of 34] {$4,5$};
    \node[ellipse, draw] (56) [right=of 45] {$5,6$};
    \node[ellipse, draw] (57) [above left =of 56] {$5,7$};
    \node[ellipse, draw] (412) [above left =of 45] {$4,12$};
    \node[ellipse, draw] (38) [above left =of 34 ] {$3,8$};
    \node[ellipse, draw] (810) [above left =of 38 ] {$8,10$};
    \node[ellipse, draw] (811) [above left =of 412 ] {$8,11$};
    \node[ellipse, draw] (29) [above left =of 23 ] {$2,9$};

    \draw (12) -- (23);
    \draw (23) -- node[above ] {$C_2$} (34);
    \draw (34) -- node[above ] {$C_3$} (45);
    \draw (45) -- node[above ] {$C_4$} (56);
    \draw (45) -- (57);
    \draw (56) -- (57);
    \draw (34) -- (412);
    \draw (45) -- (412);
    \draw (23) -- (38);
    \draw (38) --  (34);
    \draw (29) --  (23);
    \draw (810) --  node[below ] {$C_5$} (811);
    \draw (38) --  (810);
    \draw (38) --  (811);

\end{tikzpicture}
\caption{The graph $G^{1}_{1}$.}
\label{DelG1}
\end{figure}
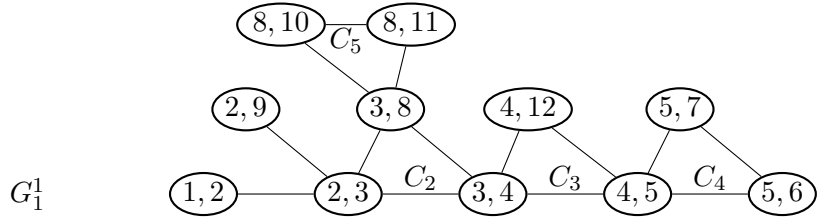

\begin{itemize}
    \item[2.2.] For $k=1$, $D_1=\{C_1\}$ and $S_1=\{e_1\}$.\\
    2.2.1(ii(b)). Since $ (C_2, C_1) \in (\Delta - D_1) \times D_1$ and $\deg_{G^1_1}(\{2,3\}) = 4$, then select the edge with the minimum weight in $C_2$ among those containing the vertex $\{2,3\}$, namely $e_2.$\\
    2.2.2. $D_2 = D_1 \cup \{C_2\}= \{C_1, C_2\}$ and $S_{2}=S_{1} \cup \{ e_2\}= \{e_1, e_2\}.$\\
    2.2.3. Set $E(G^1_{2}) = E(L(T_1)) - S_{2}.$
\end{itemize}

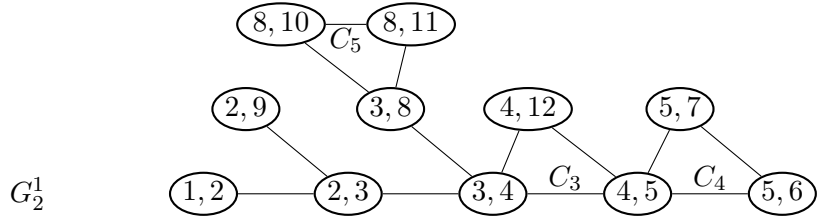
\begin{figure}[htbp]
\centering
\begin{tikzpicture}[every node/.style={thick, inner sep=1pt, font=\normalsize}, node distance=1 cm]

    % T1 layer
    \node (T1) at (-1,-1.75) {$G^{1}_{2}$};
    \node[ellipse, draw] (23) [right=3.5 cm of T1] {$2,3$};
    \node[ellipse, draw] (12) [left=of 23] {$1,2$};
    \node[ellipse, draw] (34) [right=of 23] {$3,4$};
    \node[ellipse, draw] (45) [right=of 34] {$4,5$};
    \node[ellipse, draw] (56) [right=of 45] {$5,6$};
    \node[ellipse, draw] (57) [above left =of 56] {$5,7$};
    \node[ellipse, draw] (412) [above left =of 45] {$4,12$};
    \node[ellipse, draw] (38) [above left =of 34 ] {$3,8$};
    \node[ellipse, draw] (810) [above left =of 38 ] {$8,10$};
    \node[ellipse, draw] (811) [above left =of 412 ] {$8,11$};
    \node[ellipse, draw] (29) [above left =of 23 ] {$2,9$};

    \draw (12) --  (23);
    \draw (23) --  (34);
    \draw (34) -- node[above ] {$C_3$} (45);
    \draw (45) -- node[above ] {$C_4$} (56);
    \draw (45) -- (57);
    \draw (56) -- (57);
    \draw (34) -- (412);
    \draw (45) -- (412);
    \draw (38) --  (34);
    \draw (29) --  (23);
    \draw (810) --  node[below ] {$C_5$} (811);
    \draw (38) --  (810);
    \draw (38) --  (811);
    
\end{tikzpicture}
\caption{The graph $G^{1}_{2}$.}
\label{DelG2}
\end{figure}

\begin{itemize}
    \item[2.2.] For $k=2$, $D_2= \{C_1, C_2\}$ and $S_{2}= \{e_1, e_2\}.$\\ 
    2.2.1(ii(b)). Since $(C_3, C_2) \in (\Delta - D_2) \times D_2$ and $\deg_{G^1_2}(\{3,4\}) = 4$, then select the edge with the minimum weight in $C_3$ among those containing the vertex $\{3,4\}$, namely $e_3$.\\
    2.2.2. $D_3 = D_2 \cup \{C_3\}= \{C_1, C_2, C_3\}$ and $S_{3}=S_{2} \cup \{ e_3\}= \{e_1, e_2, e_3\}.$\\
    2.2.3. Set $E(G^1_{3}) = E(L(T_1)) - S_{3}.$
\end{itemize}

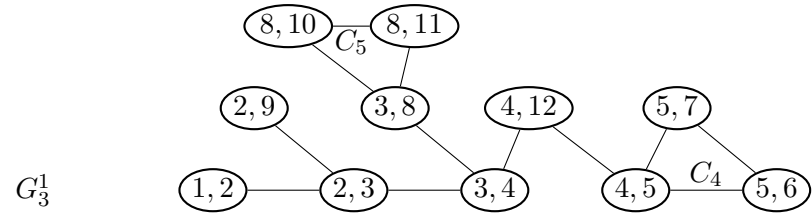
\begin{figure}[htbp]
\centering
\begin{tikzpicture}[every node/.style={thick, inner sep=1pt, font=\normalsize}, node distance=0.95 cm]

    % T1 layer
    \node (T1) at (-1,-1.75) {$G^{1}_{3}$};
    \node[ellipse, draw] (23) [right=3.5 cm of T1] {$2,3$};
    \node[ellipse, draw] (12) [left=of 23] {$1,2$};
    \node[ellipse, draw] (34) [right=of 23] {$3,4$};
    \node[ellipse, draw] (45) [right=of 34] {$4,5$};
    \node[ellipse, draw] (56) [right=of 45] {$5,6$};
    \node[ellipse, draw] (57) [above left =of 56] {$5,7$};
    \node[ellipse, draw] (412) [above left =of 45] {$4,12$};
    \node[ellipse, draw] (38) [above left =of 34 ] {$3,8$};
    \node[ellipse, draw] (810) [above left =of 38 ] {$8,10$};
    \node[ellipse, draw] (811) [above left =of 412 ] {$8,11$};
    \node[ellipse, draw] (29) [above left =of 23 ] {$2,9$};

    \draw (12) --  (23);
    \draw (23) -- (34);
    \draw (45) -- node[above ] {$C_4$} (56);
    \draw (45) -- (57);
    \draw (56) -- (57);
    \draw (34) -- (412);
    \draw (45) -- (412);
    \draw (38) --  (34);
    \draw (29) --  (23);
    \draw (810) --  node[below ] {{\color{black}$C_5$}} (811);
    \draw (38) --  (810);
    \draw (38) --  (811);
    
\end{tikzpicture}
\caption{The graph $G^{1}_{3}$.}
\label{DelG3}
\end{figure}

\newpage
\begin{itemize}
    \item[2.2.] For $k=3$, $D_3 =\{C_1, C_2, C_3\}$ and $S_{3}=\{e_1, e_2, e_3\}.$\\ 
    2.2.1(ii(a)). Since $(C_4, C_3) \in (\Delta - D_3) \times D_3$ and $\deg_{G^1_3}(\{4,5\})=3 < 4$, then select the edge in $C_4$ with the minimum weight, namely $e_4$.\\
    2.2.2. $D_4 = D_3 \cup \{C_4\}= \{C_1, C_2, C_3, C_4\}$ and $S_{4}=S_{3} \cup \{ e_4\}= \{e_1, e_2, e_3, e_4\}.$\\
    2.2.3. Set $E(G^1_{4}) = E(L(T_1)) - S_{4}.$
\end{itemize}

\begin{figure}[htbp]
\centering
\begin{tikzpicture}[every node/.style={thick, inner sep=1pt, font=\normalsize}, node distance=1 cm]

    % T1 layer
    \node (T1) at (-1,-1.75) {$G^{1}_{4}$};
    \node[ellipse, draw] (23) [right=3.5 cm of T1] {$2,3$};
    \node[ellipse, draw] (12) [left=of 23] {$1,2$};
    \node[ellipse, draw] (34) [right=of 23] {$3,4$};
    \node[ellipse, draw] (45) [right=of 34] {$4,5$};
    \node[ellipse, draw] (56) [right=of 45] {$5,6$};
    \node[ellipse, draw] (57) [above left =of 56] {$5,7$};
    \node[ellipse, draw] (412) [above left =of 45] {$4,12$};
    \node[ellipse, draw] (38) [above left =of 34 ] {$3,8$};
    \node[ellipse, draw] (810) [above left =of 38 ] {$8,10$};
    \node[ellipse, draw] (811) [above left =of 412 ] {$8,11$};
    \node[ellipse, draw] (29) [above left =of 23 ] {$2,9$};

    \draw (12) --  (23);
    \draw (23) --  (34);
    \draw (45) --  (56);
    \draw (45) -- (57);
    \draw (34) -- (412);
    \draw (45) -- (412);
    \draw (38) --  (34);
    \draw (29) --  (23);
    \draw (810) --  node[below ] {{\color{black}$C_5$}} (811);
    \draw (38) --  (810);
    \draw (38) --  (811);
    
\end{tikzpicture}
\caption{The graph $G^{1}_{4}$.}
\label{DelG4}
\end{figure}
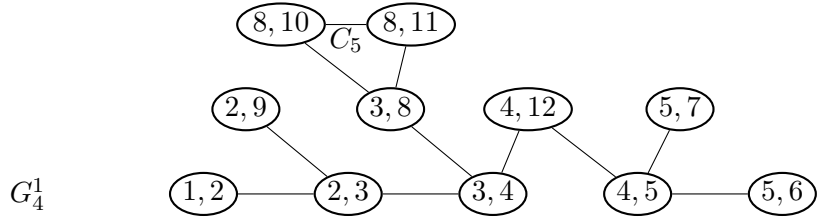
\begin{itemize}
    \item[ 2.2.] For $k=4$, $D_4 = \{C_1, C_2, C_3, C_4\}$ and $S_{4}=\{e_1, e_2, e_3, e_4\}.$\\  
    2.2.1(ii(a)).  Since $ (C_5, C_2) \in (\Delta - D_4) \times D_4$ and $\deg_{G^1_4}(\{3,8\})=3 < 4$, then select the edge in $C_5$ with the minimum weight, namely $e_5$.\\
    2.2.2. $D_5 = D_4 \cup \{C_5\}= \{C_1, C_2, C_3, C_4, C_5\}$ and $S_{5}=S_{4} \cup \{ e_5\}= \{e_1, e_2, e_3, e_4, e_5\}.$\\
    2.2.3. Set $E(G^1_{5}) = E(L(T_1)) - S_{5}.$
\end{itemize}
\begin{comment}
    \begin{figure}[htbp]
\centering
\begin{tikzpicture}[every node/.style={thick, inner sep=1pt, font=\normalsize}, node distance=1.5 cm]

    % T1 layer
    \node (T1) at (-1,-1.75) {$G^{1}_{5}=T_2$};
    \node[ellipse, draw] (23) [right=3.5 cm of T1] {$2,3$};
    \node[ellipse, draw] (12) [left=of 23] {$1,2$};
    \node[ellipse, draw] (34) [right=of 23] {$3,4$};
    \node[ellipse, draw] (45) [right=of 34] {$4,5$};
    \node[ellipse, draw] (56) [right=of 45] {$5,6$};
    \node[ellipse, draw] (57) [above left =of 56] {$5,7$};
    \node[ellipse, draw] (412) [above left =of 45] {$4,12$};
    \node[ellipse, draw] (38) [above left =of 34 ] {$3,8$};
    \node[ellipse, draw] (810) [above left =of 38 ] {$8,10$};
    \node[ellipse, draw] (811) [above left =of 412 ] {$8,11$};
    \node[ellipse, draw] (29) [above left =of 23 ] {$2,9$};

    \draw (12) -- node[above ] {$C_1$} (23);
    \draw [red, thick] (12) -- (29);
    \draw (23) -- node[above ] {$C_2$} (34);
    \draw [teal, thick] (34) -- node[above ] {{\color{black}$C_3$}} (45);
    \draw (45) -- node[above ] {$C_4$} (56);
    \draw (45) -- (57);
    \draw [violet, thick] (56) -- (57);
    \draw (34) -- (412);
    \draw (45) -- (412);
    \draw [blue, thick] (23) -- (38);
    \draw (38) --  (34);
    \draw (29) --  (23);
    \draw [orange, thick] (810) --  node[below ] {{\color{black}$C_5$}} (811);
    \draw (38) --  (810);
    \draw (38) --  (811);
\end{tikzpicture}
\end{figure}
\end{comment}

\begin{figure}[htbp]
\centering
\begin{tikzpicture}[every node/.style={thick, inner sep=1pt, font=\normalsize}, node distance=1 cm]
    
   % T1 layer
    \node (T1) at (-1,-1.75) {$G^{1}_{5}=T_2$};
    \node[ellipse, draw] (23) [right=3.5 cm of T1] {$2,3$};
    \node[ellipse, draw] (12) [left=of 23] {$1,2$};
    \node[ellipse, draw] (34) [right=of 23] {$3,4$};
    \node[ellipse, draw] (45) [right=of 34] {$4,5$};
    \node[ellipse, draw] (56) [right=of 45] {$5,6$};
    \node[ellipse, draw] (57) [above left =of 56] {$5,7$};
    \node[ellipse, draw] (412) [above left =of 45] {$4,12$};
    \node[ellipse, draw] (38) [above left =of 34 ] {$3,8$};
    \node[ellipse, draw] (810) [above left =of 38 ] {$8,10$};
    \node[ellipse, draw] (811) [above left =of 412 ] {$8,11$};
    \node[ellipse, draw] (29) [above left =of 23 ] {$2,9$};
    
    \draw (12) -- (23);
    \draw (23) -- (34);
    \draw (45) -- (56);
    \draw (45) -- (57);
    \draw (34) -- (412);
    \draw (45) -- (412);
    \draw (38) --  (34);
    \draw (29) --  (23);
    \draw (38) --  (810);
    \draw (38) --  (811);
    
\end{tikzpicture}
\caption{The graph $G^{1}_{5}$.}
\label{DelG5}
\end{figure}
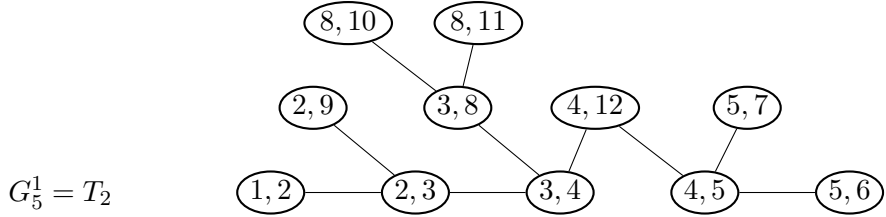

For \(n = 2, \ldots, 11\), Definition~\ref{def:  delta tree seq} can be applied analogously to Example~\ref{ex: v12}. Moreover, by Remark~\ref{rem 4.4}, every \(\Delta\)-vine tree sequence forms an R-vine tree sequence.
\subsection{Algorithm for $\Delta$-Vine}\label{Algdet}
As in Section~\ref{AlgSL}, Steps~2 and~12 of Di{\ss}mann’s Algorithm must be adapted. To apply the \(\Delta\)-vine sequence in Definition~\ref{def:  delta tree seq} to real datasets, we follow the tree-selection procedure of Step~2 in Di{\ss}mann’s Algorithm while imposing the additional restriction that the degree of every node is at most three. Specifically, we define
\begin{align*}\label{spanning}
    T_1 = \underset{
\substack{
T=(N,E)\ \text{spanning tree} \\
\deg(v) \leq 3\ \forall v \in N
}
}{\arg\max} \sum_{e = (a_e, b_e) \in E} \tau_{a_e, b_e}.
\end{align*}The construction of $T_{2}, \dots, T_{d-1}$ is presented in Algorithm~\ref{GetToDelete} in Appendix \ref{app:AlgDel}.

%%%%%%%%%%%%%%%%%%%%%%%%%%
\section{Simulating from SL-vine and $\Delta$-vine copulas}\label{sec5}
This section presents examples illustrating the use of the \texttt{VineCopula} package, focusing on simulations from the SL-vine and \(\Delta\)-vine copulas.

The function \texttt{RVineSim} from the R package \texttt{VineCopula} can be used to generate samples from a specified R-vine copula model. Since both the SL-vine and the \(\Delta\)-vine are subclasses of the R-vine, this function can also be applied directly to simulate observations from SL-vine and \(\Delta\)-vine copulas. For the 8-dimensional case, we represent the SL-vine and \(\Delta\)-vine tree structures by the following lower triangular matrices \(M_{\text{SL}}\) and \(M_{\Delta}\), respectively. That is,
\begin{center}
    $M_{\text{SL}}=\begin{pmatrix}
6 & 0 & 0 & 0 & 0 & 0 & 0 & 0 \\
3 & 7 & 0 & 0 & 0 & 0 & 0 & 0 \\
7 & 3 & 2 & 0 & 0 & 0 & 0 & 0 \\
5 & 5 & 3 & 3 & 0 & 0 & 0 & 0 \\
8 & 8 & 5 & 5 & 1 & 0 & 0 & 0 \\
4 & 4 & 8 & 8 & 5 & 5 & 0 & 0 \\
1 & 1 & 4 & 4 & 8 & 8 & 4 & 0 \\
2 & 2 & 1 & 1 & 4 & 4 & 8 & 8
\end{pmatrix}$ and 
$M_{\Delta}=\begin{pmatrix}
5 & 0 & 0 & 0 & 0 & 0 & 0 & 0 \\
3 & 6 & 0 & 0 & 0 & 0 & 0 & 0 \\
8 & 3 & 4 & 0 & 0 & 0 & 0 & 0 \\
7 & 8 & 3 & 1 & 0 & 0 & 0 & 0 \\
2 & 7 & 8 & 3 & 3 & 0 & 0 & 0 \\
1 & 2 & 7 & 8 & 8 & 2 & 0 & 0 \\
6 & 1 & 2 & 7 & 2 & 8 & 7 & 0 \\
4 & 4 & 1 & 2 & 7 & 7 & 8 & 8
\end{pmatrix}$.
\end{center}

As pair-copula families, we selected
\begin{center}
$\text{fam}_{\text{SL}}=\begin{pmatrix}
0 & 0 & 0 & 0 & 0 & 0 & 0 & 0 \\
5 & 0 & 0 & 0 & 0 & 0 & 0 & 0 \\
5 & 2 & 0 & 0 & 0 & 0 & 0 & 0 \\
20 & 5 & 114 & 0 & 0 & 0 & 0 & 0 \\
2 & 134 & 1 & 30 & 0 & 0 & 0 & 0 \\
5 & 14 & 7 & 7 & 2 & 0 & 0 & 0 \\
1 & 17 & 134 & 1 & 2 & 2 & 0 & 0 \\
14 & 14 & 214 & 114 & 114 & 14 & 18 & 0
\end{pmatrix}$ and
    $\text{fam}_{\Delta}=\begin{pmatrix}
0 & 0 & 0 & 0 & 0 & 0 & 0 & 0 \\
134 & 0 & 0 & 0 & 0 & 0 & 0 & 0 \\
40 & 20 & 0 & 0 & 0 & 0 & 0 & 0 \\
5 & 104 & 5 & 0 & 0 & 0 & 0 & 0 \\
30 & 134 & 30 & 0 & 0 & 0 & 0 & 0 \\
114 & 33 & 2 & 5 & 13 & 0 & 0 & 0 \\
224 & 20 & 2 & 2 & 204 & 40 & 0 & 0 \\
214 & 214 & 114 & 14 & 14 & 214 & 214 & 0
\end{pmatrix},$
\end{center}

where 0 corresponds to independence, 1 to a Gaussian, 2 to a Student t, 5 to a Frank, 7 to a BB1, 13 to a survival Clayton, 14 to a survival Gumbel, 17 to a survival BB1, 18 to a survival BB6, 20 to a survival BB8, 30 to a rotated $90^\circ$ BB8, 33 to a rotated $270^\circ$ Clayton, 40 to a rotated $270^\circ$ BB8, 104 to a Tawn type 1, 114 to a rotated $180^\circ$ Tawn type 1, 134 to a rotated $270^\circ$ Tawn type 1, 214 to a rotated $180^\circ$ Tawn type 2, 224 to a rotated $90^\circ$ Tawn type 2. The corresponding copula parameters were specified as

\begin{center}
\tiny   
$\theta_{\text{SL}}=\begin{pmatrix}
0 & 0 & 0 & 0 & 0 & 0 & 0 & 0 \\
0.41 & 0 & 0 & 0 & 0 & 0 & 0 & 0 \\
-0.62 & (-0.09, 15.51) & 0 & 0 & 0 & 0 & 0 & 0 \\
(1.10, 0.99) & 0.89 & (1.51,0.01) & 0 & 0 & 0 & 0 & 0 \\
(0.19,22.11) & (-1.25,0.06) & 0.08 & (-1.35,-0.77) & 0 & 0 & 0 & 0 \\
-1.84 & 1.10 & (0.09,1.14) & (0.06,1.11) & (0.68,5.78) & 0 & 0 & 0 \\
0.38 & (0.27,1.24) & (-1.30,0.18) & -0.55 & (0.88,4.11) & (0.47,11.29) & 0 & 0 \\
2.46 & 6.42 & (5.72,0.97) & (1.81,0.32) & (5.91,0.98) & 3.64 & (1.60,2.16) & 0
\end{pmatrix}$
\end{center} and
\begin{center}
\tiny
$\theta_{\Delta}=\begin{pmatrix}
0 & 0 & 0 & 0 & 0 & 0 & 0 & 0 \\
(-1.14,0.09) & 0 & 0 & 0 & 0 & 0 & 0 & 0 \\
(-1.80,-0.89) & (1.15,0.98) & 0 & 0 & 0 & 0 & 0 & 0 \\
-5.65 & (1.12,0.16) & 1.22 & 0 & 0 & 0 & 0 & 0 \\
(-1.68,-0.69) & (-1.38,0.34) & (-1.35,-0.87) & 0 & 0 & 0 & 0 & 0 \\
(1.33,0.16) & -0.09 & (0.53,24.91) & -1.33 & 0.05 & 0 & 0 & 0 \\
(-1.64,0.22) & (1.18,0.94) & (0.39,6.27) & (0.18,8.79) &
(1.23,0.16) & (-1.50,-0.95) & 0 & 0 \\
(5.91,0.98) & (5.09,0.99) & (5.76,0.98) & 6.42 &
2.83 & (4.57,0.95) & (1.92,0.44) & 0
\end{pmatrix}.$
\end{center}

The corresponding SL-vine and \(\Delta\)-vine tree sequences are displayed in Figure~\ref{monggon:SvsDelta}. Samples of size \(n=1307\) were then generated from the SL-vine and \(\Delta\)-vine copulas using the function \texttt{RVineSim} in the \texttt{VineCopula} package with inputs \(M\), \texttt{fam}, and \(\theta\). The resulting pair plots and normalized contour plots are shown in Figures~\ref{SLPlot} and~\ref{DeltaPlot}, respectively.

To evaluate the performance of the proposed vine copula models, namely the SL-vine and the \(\Delta\)-vine, relative to the conventional vine copula, we employ the SL-vine algorithm from Section~\ref{AlgSL} and the \(\Delta\)-vine algorithm from Section~\ref{Algdet}. These procedures generate the tree sequences \(T_1, \dots, T_{d-1}\), which determine the corresponding copula density structures. Model performance is then assessed using the Akaike Information Criterion (AIC) and Bayesian Information Criterion (BIC). The comparison is carried out on the abalone, capital-market, and Wisconsin breast cancer datasets.
\begin{figure}[htbp]
\centering
\begin{minipage}{0.48\textwidth}
\begin{tikzpicture}[every node/.style={thick, inner sep=1pt, font=\scriptsize}, node distance=0.35cm]

    % T1 layer
    \node (T1) at (-2,0) {$T_1$};
    
    \node[ellipse, draw] (2) [right=2 cm of T1] {$2$};
    \node[ellipse, draw] (1) [left=of 2] {$1$};
    \node[ellipse, draw] (4) [right=of 2] {$4$};

    \node[ellipse, draw] (5) [right=of 4] {$5$};
    \node[ellipse, draw] (6) [right=of 5] {$6$};
    \node[ellipse, draw] (7) [above left =of 5] {$7$};
    \node[ellipse, draw] (8) [above left =of 4] {$8$};
    \node[ellipse, draw] (3) [above left =of 2] {$3$};

    \draw (1) -- (2);
    \draw (2) -- (4);
    \draw (4) -- (5);
    \draw (5) -- (6);
    \draw (5) -- (7);
    \draw (8) --  (4);
    \draw (3) --  (2);
    
    % T2 layer
    
    \node (T2) at (-2,-2) {$T_2$};
    \node[ellipse, draw] (2) [right=2.4 cm of T2] {$2,4$};
    \node[ellipse, draw] (1) [left=of 2] {$1,2$};
    \node[ellipse, draw] (4) [right=of 2] {$4,5$};
    \node[ellipse, draw] (5) [right=of 4] {$5,7$};
    \node[ellipse, draw] (7) [above left =of 5] {$5,6$};
    \node[ellipse, draw] (8) [above left =of 4] {$4,8$};
    \node[ellipse, draw] (3) [above left =of 2] {$2,3$};

    \draw (1) -- (2);
    \draw (2) -- (4);
    \draw (4) -- (5);
    \draw (5) -- (7);
    \draw (8) --  (4);
    \draw (3) --  (2);

     % T3 layer
    \node (T3) at (-2,-4) {$T_3$};
    \node[ellipse, draw] (2) [right=2.25 cm of T3] {$2,5|4$};
    \node[ellipse, draw] (1) [left=of 2] {$1,4|2$};
    \node[ellipse, draw] (4) [right=of 2] {$4,7|5$};
    \node[ellipse, draw] (5) [right=of 4] {$5,8|4$};
    \node[ellipse, draw] (8) [above left =of 4] {$6,7|5$};
    \node[ellipse, draw] (3) [above left =of 2] {$3,4|2$};

    \draw (1) -- (2);
    \draw (2) -- (4);
    \draw (4) -- (5);
    \draw (8) --  (4);
    \draw (3) --  (2);

    % T4 layer
    \node (T4) at (-2,-6) {$T_4$};
    \node[ellipse, draw] (2) [right=2.8 cm of T4] {$2,7|5,4$};
    \node[ellipse, draw] (1) [left=of 2] {$1,5|4,2$};
    \node[ellipse, draw] (4) [right=of 2] {$4,6|7,5$};
    \node[ellipse, draw] (8) [above left =of 4] {$7,8|4,5$};
    \node[ellipse, draw] (3) [above left =of 2] {$3,5|4,2$};

    \draw (1) -- (2);
    \draw (2) -- (4);
    \draw (8) --  (4);
    \draw (3) --  (2);

    % T5 layer
    \node (T5) at (-2,-8) {$T_5$};
    \node[ellipse, draw] (2) [right=2.75 cm of T5] {$2,6|7,5,4$};
    \node[ellipse, draw] (1) [left=of 2] {$1,7|5,4,2$};
    \node[ellipse, draw] (4) [right=of 2] {$8,6|7,4,5$};
    \node[ellipse, draw] (3) [above left =of 2] {$3,7|5,4,2$};

    \draw (1) -- (2);
    \draw (2) -- (4);
    \draw (3) --  (2);

    % T6 layer
    \node (T6) at (-2,-10) {$T_6$};
    \node[ellipse, draw] (2) [right=3.25 cm of T6] {$1,6|7,5,4,2$};
    \node[ellipse, draw] (1) [left=of 2] {$2,8|6,7,5,4$};
    \node[ellipse, draw] (4) [above left=of 2] {$3,6|7,5,4,2$};

    \draw (1) -- (2);
    \draw (2) -- (4);

    % T7 layer
    \node (T7) at (-2,-11) {$T_7$};
    \node[ellipse, draw] (2) [right=3.25 cm of T7] {$1,8|6,7,5,4,2$};
    \node[ellipse, draw] (1) [left=of 2] {$1,3|6,7,5,4,2$};

    \draw (1) -- (2);
\end{tikzpicture}
\end{minipage}
\hfill
\begin{minipage}{0.48\textwidth}
\centering
\begin{tikzpicture}[every node/.style={thick, inner sep=1pt, font=\scriptsize}, node distance=0.35cm]

    % T1 layer
    \node (T1) at (-2,0) {$T_1$};
    
    \node[ellipse, draw] (2) [right=2 cm of T1] {$4$};
    \node[ellipse, draw] (1) [left=of 2] {$5$};
    \node[ellipse, draw] (4) [right=of 2] {$1$};

    \node[ellipse, draw] (5) [right=of 4] {$2$};
    \node[ellipse, draw] (6) [right=of 5] {$7$};
    \node[ellipse, draw] (7) [right=of 6] {$3$};
    \node[ellipse, draw] (3) [above left =of 2] {$6$};
    \node[ellipse, draw] (8) [above left =of 6] {$8$};

    \draw (1) -- (2);
    \draw (2) -- (4);
    \draw (4) -- (5);
    \draw (5) -- (6);
    \draw (3) --  (2);
    \draw (6) --  (8);
    \draw (6) --  (7);
    
    % T2 layer
    
    \node (T2) at (-2,-2) {$T_2$};
    \node[ellipse, draw] (2) [right=1.5 cm of T2] {$4,6$};
    \node[ellipse, draw] (1) [left=of 2] {$4,5$};
    \node[ellipse, draw] (4) [above=0.25 cm of 2] {$1,4$};

    \node[ellipse, draw] (5) [right=of 4] {$1,2$};
    \node[ellipse, draw] (6) [right=of 5] {$2,7$};
    \node[ellipse, draw] (7) [right=of 6] {$3,7$};
    \node[ellipse, draw] (8) [below =0.25 cm of 6] {$7,8$};

    \draw (1) -- (2);
    \draw (2) -- (4);
    \draw (4) -- (5);
    \draw (5) -- (6);
    \draw (6) --  (8);
    \draw (6) --  (7);

     % T3 layer
    
    \node (T2) at (-2,-4) {$T_3$};
    \node[ellipse, draw] (2) [right=1.85 cm of T3] {$1,6|4$};
    \node[ellipse, draw] (1) [left=of 2] {$5,6|4$};
    \node[ellipse, draw] (4) [above=0.25 cm of 2] {$2,4|1$};

    \node[ellipse, draw] (5) [right=of 4] {$1,7|2$};
    \node[ellipse, draw] (6) [right=of 5] {$2,8|7$};
    \node[ellipse, draw] (7) [below =0.25 cm of 6] {$2,3|7$};
    
    \draw (1) -- (2);
    \draw (2) -- (4);
    \draw (4) -- (5);
    \draw (5) -- (6);
    \draw (6) --  (7);

    % T4 layer
    \node (T4) at (-2,-6) {$T_4$};
    \node[ellipse, draw] (2) [right=2.25 cm of T4] {$2,6|1,4$};
    \node[ellipse, draw] (1) [left=of 2] {$1,5|4,6$};
    \node[ellipse, draw] (4) [above=0.25 cm of 2] {$4,7|1,2$};

    \node[ellipse, draw] (5) [right=of 4] {$1,8|2,7$};
    \node[ellipse, draw] (6) [below =0.25 cm of 5] {$3,8|2,7$};
    
    \draw (1) -- (2);
    \draw (2) -- (4);
    \draw (4) -- (5);
    \draw (5) -- (6);

    % T5 layer
    \node (T5) at (-2,-8) {$T_5$};
    \node[ellipse, draw] (2) [right=3 cm of T5] {$6,7|1,2,4$};
    \node[ellipse, draw] (1) [left=of 2] {$2,5|1,4,6$};
    \node[ellipse, draw] (4) [above=0.25 cm of 2] {$4,8|1,2,7$};

    \node[ellipse, draw] (5) [left=of 4] {$1,3|2,7,8$};
    
    \draw (1) -- (2);
    \draw (2) -- (4);
    \draw (4) -- (5);

    % T6 layer
    \node (T6) at (-2,-10) {$T_6$};
    \node[ellipse, draw] (2) [right=3.25 cm of T6] {$6,8|1,2,4,7$};
    \node[ellipse, draw] (1) [left=of 2] {$5,7|1,2,4,6$};
    \node[ellipse, draw] (4) [above left=of 2] {$3,4|1,2,7,8$};

    \draw (1) -- (2);
    \draw (2) -- (4);

    % T7 layer
    \node (T7) at (-2,-11) {$T_7$};
    \node[ellipse, draw] (2) [right=3.25 cm of T7] {$3,6|1,2,4,7,8$};
    \node[ellipse, draw] (1) [left=of 2] {$5,8|1,2,4,6,7$};

    \draw (1) -- (2);
\end{tikzpicture}
\end{minipage}
\caption{Tree structures of the 8-dimensional SL-vine (left) and \(\Delta\)-vine (right).}
    \label{monggon:SvsDelta}
\end{figure}
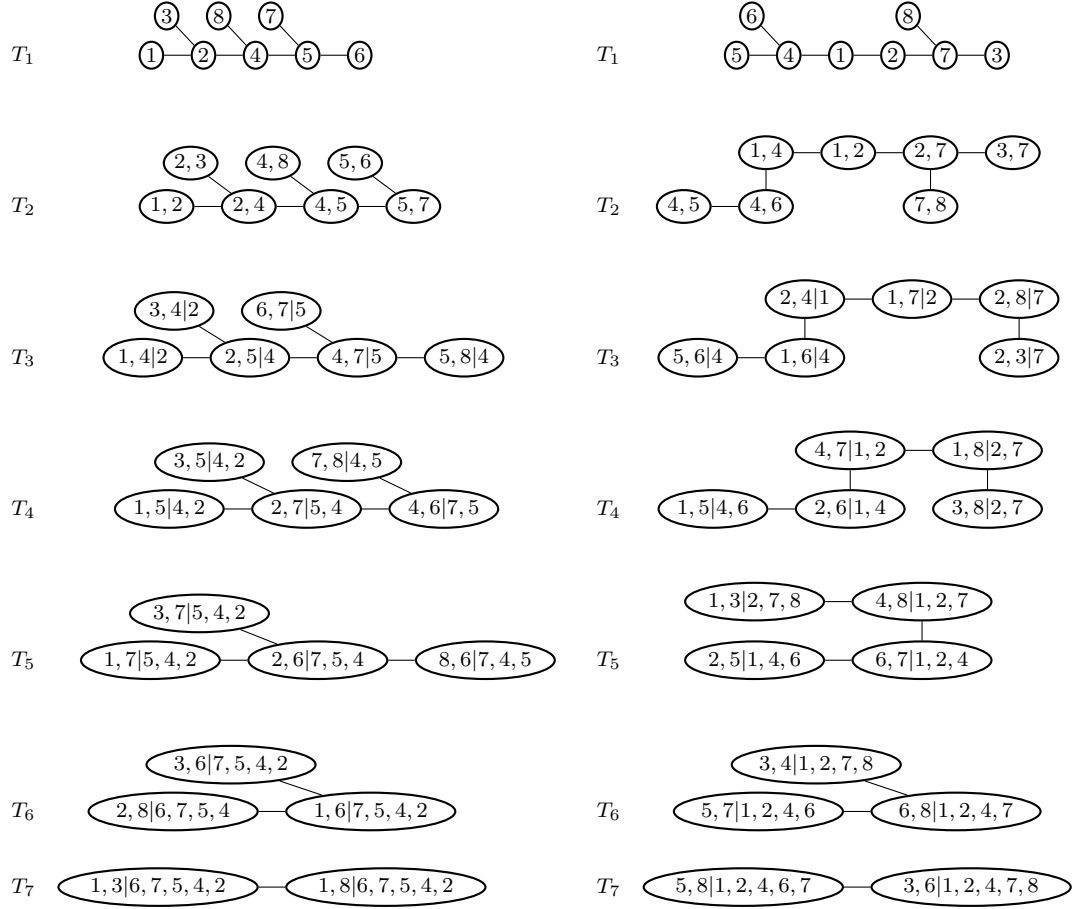 

\begin{figure}
        \centering
		\includegraphics[scale=0.25]{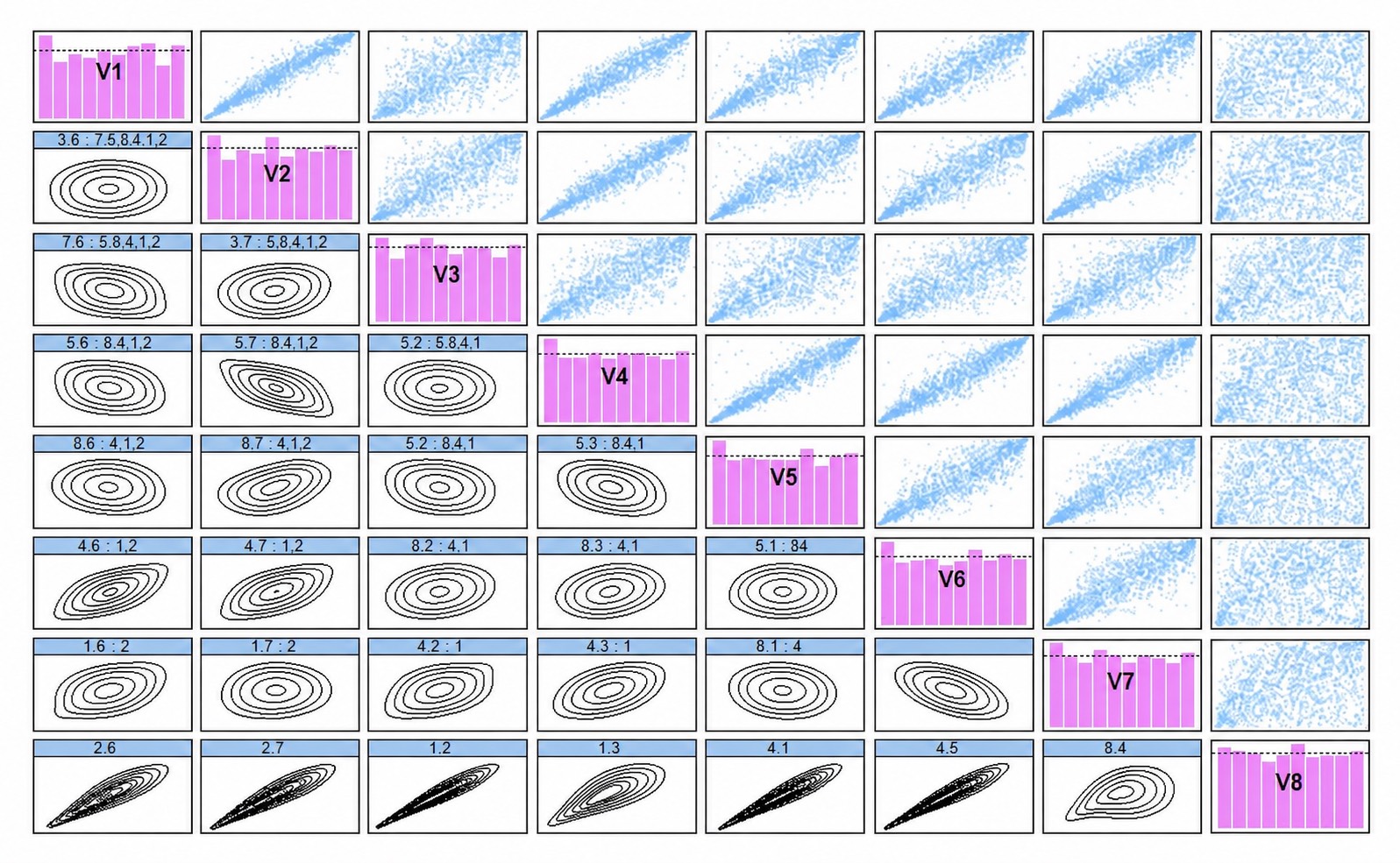}
	\caption{ Pairwise scatter plots (upper triangle), marginal histograms (diagonal), and pairwise normalized contour plots (lower triangle) based on 1307 simulated observations from the SL-vine copula.}
	\label{SLPlot}
\end{figure}

\begin{figure}
        \centering
		\includegraphics[scale=0.24]{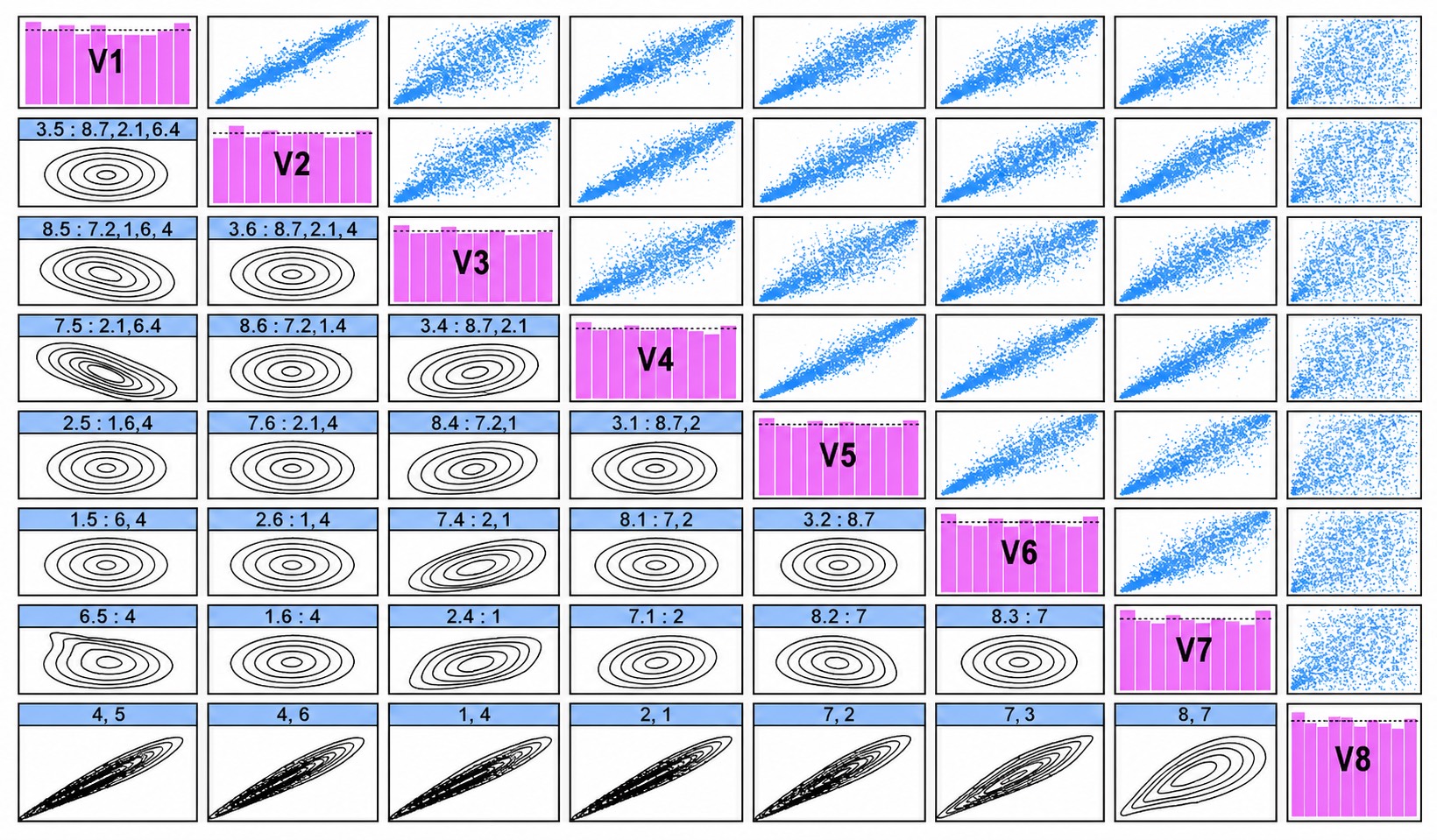}
	\caption{ Pairwise scatter plots (upper triangle), marginal histograms (diagonal), and pairwise normalized contour plots (lower triangle) based on 1307 simulated observations from the \(\Delta\)-vine copula.}
	\label{DeltaPlot}
\end{figure}
\newpage

\section{Data and result}\label{sec6}
To demonstrate the practical application of the SL-vine and \(\Delta\)-vine copula models, we consider three benchmark datasets: the Abalone, Capital Market, and Breast Cancer Wisconsin datasets. The proposed models are evaluated against existing vine copula models using the log-likelihood, Akaike Information Criterion (AIC), and Bayesian Information Criterion (BIC). Section~\ref{SecAba} examines the Abalone dataset and provides an example in which the SL-vine achieves the best overall performance. Section~\ref{SecThai} highlights the impact of the model selection criterion on the final choice of model; the Capital Market dataset illustrates that the log-likelihood, AIC, and BIC may favor different vine structures. Finally, Section~\ref{SecBCW} analyzes the Breast Cancer Wisconsin dataset, where the \(\Delta\)-vine and R-vine yield identical fits and are consistently selected as the preferred models under all three criteria. All analyses were conducted in \texttt{R} using the \texttt{VineCopula} and \texttt{rvinecopulib} packages developed by \cite{nagler2024solving}.
\subsection{Abalone dataset}\label{SecAba}
We first consider the Abalone dataset from the University of California--Irvine Machine Learning Repository (\url{https://doi.org/10.24432/C55C7W}). Following a common preprocessing step, we restrict attention to female abalones, yielding a sample of 1,307 observations. The variables analyzed are length (len), diameter (dia), height (hei), whole weight (whole), shucked weight (shuck), viscera weight (vis), shell weight (shell), and rings. Here, length, diameter, and height measure shell size (mm), whereas whole weight, shucked weight, viscera weight, and shell weight represent weight-related characteristics (g). The variable rings is used as an indicator of age, with age estimated by adding 1.5 to the ring count. Table~\ref{monggon:tab0} reports summary statistics for the selected variables.

\begin{table}[htbp]
	\centering
	\caption{Summary statistics of selected variables in the abalone dataset.}
	\label{monggon:tab0}
	\footnotesize
	\begin{tabular}{l|ccccccc}
		\hline
		Variable & Mean & SD & Median & Min & Max & Skew & Kurtosis \\
		\hline
		Length (1)        & 0.58  & 0.09 & 0.59  & 0.28 & 0.81 & -0.53 & 0.10 \\
		Diameter (2)      & 0.45  & 0.07 & 0.47  & 0.20 & 0.65 & -0.51 & 0.15 \\
		Height (3)        & 0.16  & 0.04 & 0.16  & 0.01 & 1.13 & 10.90 & 265.19 \\
		Whole weight (4)  & 1.05  & 0.43 & 1.04  & 0.08 & 2.66 & 0.37 & 0.05 \\
		Shucked weight (5) & 0.45  & 0.20 & 0.44  & 0.03 & 1.49 & 0.55 & 0.65 \\
		Viscera weight (6) & 0.23  & 0.10 & 0.22  & 0.02 & 0.59 & 0.39 & -0.10 \\
		Shell weight (7)  & 0.30  & 0.13 & 0.30  & 0.03 & 0.98 & 0.69 & 1.40 \\
		Rings (8)         & 11.13 & 3.10 & 10.00 & 5.00 & 29.00 & 1.47 & 3.10 \\
		\hline
	\end{tabular}
\end{table}
Pseudo-observations were obtained nonparametrically using the empirical cumulative distribution function. Based on these pseudo-observations, the SL-vine and $\Delta$-vine copula models were fitted using Di{\ss}mann’s algorithm together with the selection procedures described in Sections~\ref{AlgSL} and~\ref{Algdet}, respectively. The resulting tree structures are shown in Figure \ref{monggon:SvsDelta}, while the selected pair-copula families and corresponding parameter estimates are reported in Tables~\ref{monggon:tab1} and~\ref{monggon:tab1del}.

\begin{table}[htbp]
	\centering
	\caption{The SL-vine-based results present sequential parameter estimation, chosen copula structures, and corresponding dependence measures for eight variables in the female abalone data.}
	\label{monggon:tab1}
	\footnotesize
	\begin{tabular}{l|llll}
		\hline
		\textbf{Tree} & \textbf{Edge\textsuperscript{a}} & \textbf{Family} & \textbf{Parameter(s)} & \textbf{Kendall's $\tau$} \\
		\hline
		1 & $(2,3)$       & Survival Gumbel        & $2.46$        & $0.59$ \\
		1 & $(2,1)$       & Survival Gumbel        & $6.42$        & $0.84$ \\
		1 & $(4,2)$       & Tawn type 2 $(180^\circ)$ & $(5.72, 0.97)$ & $0.81$ \\
		1 & $(4,8)$       & Tawn type 1 $(180^\circ)$   & $(1.81, 0.32)$  & $0.20$ \\
		1 & $(5,4)$       & Tawn type 1 $(180^\circ)$   & $(5.91, 0.98)$  & $0.82$ \\
		1 & $(5,6)$       & Survival Gumbel        & $3.64$          & $0.73$ \\
		1 & $(7,5)$       & Survival BB6            & $(1.60, 2.16)$  & $0.65$ \\
		\hline
		2 & $(4,3\,\mid\,2)$ & Gaussian         & $0.38$  & $0.25$ \\
		2 & $(4,1\,\mid\,2)$ & Survival BB1      & $(0.27, 1.24)$  & $0.29$ \\
		2 & $(5,2\,\mid\,4)$ & Tawn type 1 $(270^\circ)$   & $(-1.30, 0.18)$  & $-0.08$ \\
		2 & $(5,8\,\mid\,4)$ & Gaussian         & $-0.55$ & $-0.37$ \\
		2 & $(7,4\,\mid\,5)$ & Student's $t$         & $(0.88, 4.11)$  & $0.68$ \\
		2 & $(7,6\,\mid\,5)$ & Student's $t$         & $(0.47, 11.29)$ & $0.31$ \\
		\hline
		3 & $(5,3\,\mid\,4,2)$ & Frank       & $-1.84$   & $-0.20$ \\
		3 & $(5,1\,\mid\,4,2)$ & Survival Gumbel      & $1.10$    & $0.09$ \\
		3 & $(7,2\,\mid\,5,4)$ & BB1     & $(0.09, 1.14)$  & $0.16$ \\
		3 & $(7,8\,\mid\,5,4)$ & BB1     & $(0.06, 1.11)$  & $0.13$ \\
		3 & $(6,4\,\mid\,7,5)$ & Student's $t$       & $(0.68, 5.78)$  & $0.48$ \\
		\hline
		4 & $(7,3\,\mid\,5,4,2)$ & Student's $t$     & $(0.19, 22.11)$ & $0.12$ \\
		4 & $(7,1\,\mid\,5,4,2)$ & Tawn type 1 $(270^\circ)$ & $(-1.25, 0.06)$ & $-0.03$ \\
		4 & $(6,2\,\mid\,7,5,4)$ & Gaussian     & $0.08$  & $0.05$ \\
		4 & $(6,8\,\mid\,7,5,4)$ & BB8 $(90^\circ)$ & $(-1.35, -0.77)$ & $-0.07$ \\
		\hline
		5 & $(6,3\,\mid\,7,5,4,2)$ & Survival BB8 & $(1.10, 0.99)$  & $0.05$ \\
		5 & $(6,1\,\mid\,7,5,4,2)$ & Frank    & $0.89$ & $0.10$ \\
		5 & $(8,2\,\mid\,6,7,5,4)$ & Tawn type 1 $(180^\circ)$ & $(1.51, 0.01)$ & $0.01$ \\
		\hline
		6 & $(1,3\,\mid\,6,7,5,4,2)$ & Frank  & $-0.62$   & $-0.07$ \\
		6 & $(8,1\,\mid\,6,7,5,4,2)$ & Student's $t$  & $(-0.09, 15.51)$ & $-0.06$ \\
		\hline
		7 & $(8,3\,\mid\,1,6,7,5,4,2)$ & Frank & $(0.41,0.00)$ & $0.05$ \\
		\hline
	\end{tabular}
\end{table}

\begin{table}[htbp]
\centering
\caption{The $\Delta$-vine-based results present sequential parameter estimation, chosen copula structures, and corresponding dependence measures for eight variables in the female abalone data.}

\label{monggon:tab1del}
\footnotesize
\begin{tabular}{l|llll}
\hline
\textbf{Tree} & \textbf{Edge\textsuperscript{a}} & \textbf{Family} & \textbf{Parameter(s)} & \textbf{Kendall's $\tau$} \\
\hline
1 & $(4,5)$ & Tawn type 2 $(180^\circ)$ & (5.91, 0.98) & 0.82 \\
1 & $(4,6)$ & Tawn type 2 $(180^\circ)$ & (5.09, 0.99) & 0.80 \\
1 & $(1,4)$ & Tawn type 1 $(180^\circ)$ & (5.76, 0.98) & 0.81 \\
1 & $(2,1)$ & Survival Gumbel & 6.42 & 0.84 \\
1 & $(7,3)$ & Survival Gumbel & 2.83 & 0.65 \\
1 & $(7,2)$ & Tawn type 2 $(180^\circ)$ & (4.57, 0.95) & 0.75 \\
1 & $(8,7)$ & Tawn type 2 $(180^\circ)$ & (1.92, 0.44) & 0.27 \\
\hline
2 & $(6,5\,\mid\,4)$ & Tawn type 2 $(90^\circ)$ & (-1.64, 0.22) & -0.14 \\
2 & $(1,6\,\mid\,4)$ & Survival BB8 & (1.18, 0.94) & 0.06 \\
2 & $(2,4\,\mid\,1)$ & Student’s $t$ & (0.39, 6.27) & 0.25 \\
2 & $(7,1\,\mid\,2)$ & Student’s $t$ & (0.18, 8.79) & 0.11 \\
2 & $(2,3\,\mid\,7)$ & Tawn type 2 & (1.23, 0.16) & 0.06 \\
2 & $(8,2\,\mid\,7)$ & BB8 $(270^\circ)$ & (-1.50, -0.95) & -0.17 \\
\hline
3 & $(1,5\,\mid\,6,4)$ & Tawn type 1 $(180^\circ)$ & (1.33, 0.16) & 0.08 \\
3 & $(2,6\,\mid\,1,4)$ & Clayton $(270^\circ)$ & -0.09 & -0.04 \\
3 & $(7,4\,\mid\,2,1)$ & Student’s $t$ & (0.53, 24.91) & 0.36 \\
3 & $(8,1\,\mid\,7,2)$ & Frank & -1.33 & -0.14 \\
3 & $(8,3\,\mid\,2,7)$ & Survival Clayton & 0.05 & 0.03 \\
\hline
4 & $(2,5\,\mid\,1,6,4)$ & BB8 $(90^\circ)$ & (-1.68, -0.69) & -0.10 \\
4 & $(7,6\,\mid\,2,1,4)$ & Tawn type 1 $(270^\circ)$ & (-1.38, 0.34) & -0.14 \\
4 & $(8,4\,\mid\,7,2,1)$ & BB8 $(90^\circ)$ & (-1.35, -0.87) & 0.10 \\
4 & $(3,1\,\mid\,8,7,2)$ & Independent & - & 0.00 \\
\hline
5 & $(7,5\,\mid\,2,1,6,4)$ & Frank & -5.65 & -0.49 \\
5 & $(3,6\,\mid\,7,2,1,4)$ & Tawn type 1 & (-1.12, 0.16) & 0.04 \\
5 & $(3,4\,\mid\,8,7,2,1)$ & Frank & 1.22 & 0.13 \\
\hline
6 & $(8,5\,\mid\,7,2,1,6,4)$ & BB8 $(270^\circ)$ & (-1.80, -0.89) & -0.21 \\
6 & $(3,6\,\mid\,8,7,2,1,4)$ & Survival BB8 & (1.15, 0.98) & 0.07 \\
\hline
7 & $(3,5\,\mid\,8,7,2,1,6,4)$ & Tawn type 1 $(270^\circ)$ & (-1.14, 0.09) & -0.03 \\
\hline
\end{tabular}
\end{table}

For comparison, we additionally fit R-vine, C-vine, D-vine, and Gaussian vine copula models to the Abalone dataset. The Gaussian vine uses the same tree structure as the R-vine but restricts all pair copulas to the Gaussian family. Table~\ref{monggon:tab2} reports the log-likelihood, AIC, and BIC values for all models. According to both AIC and BIC, the SL-vine provides the best fit, followed by the C-vine, R-vine, $\Delta$-vine, D-vine, and Gaussian vine. This example demonstrates that the proposed SL-vine can outperform existing vine copula models.

\begin{table}[htbp]
	\centering
	\caption{Comparison of copula models on the abalone dataset using log-Likelihood, AIC, and BIC.}
	\small
	\begin{tabular}{l|ccc}
		\hline
		Copula model & Log-likelihood & AIC & BIC \\ \hline
		SL-vine copula        & $10,444.72$ & $-20,799.44$ & $-20,566.54$ \\ 
		C-vine copula         & $10,421.85$ & $-20,747.70$ & $-20,499.27$ \\ 
		R-vine copula         & $10,411.47$ & $-20,730.93$ & $-20,492.86$ \\ 
		$\Delta$-vine copula  & $10,328.65$ & $-20,563.31$ & $-20,320.06$ \\ 
		D-vine copula         & $9,905.04$  & $-19,720.08$ & $-19,487.18$ \\ 
		Gaussian vine copula  & $9,363.33$  & $-18,670.66$ & $-18,525.75$ \\ \hline
	\end{tabular}
	\label{monggon:tab2}
\end{table}
\subsection{Capital-Market dataset}\label{SecThai}
The second dataset considered is the Capital Market dataset of \cite{Nacaskul2023MultiParadigm}, which contains 14 financial variables representing domestic (Thailand), regional (Asia/Emerging Markets), and international (United States) market activity. The dataset comprises 3,389 daily observations from January~2009 to December~2021 and was compiled from the Stock Exchange of Thailand (SET), the Thai Bond Market Association (Thai BMA), and Bloomberg. For descriptions of the variables, the dataset includes domestic equity returns (SET and MAI), short- and long-term government bond yields (ZeroShort and ZeroLong), corporate bond spreads (CorpBond), exchange-rate returns (THB), foreign capital flows (EquityFlow and BondFlow), regional and global equity and bond returns (EMAsiaEquity, SP500, EMBond, and USBond), and regional and U.S. dollar currency indices (EMAsiaFX and USD). Table~\ref{tab:desc-stats-full} reports summary statistics for the financial variables.

\begin{table}[htbp]
	\centering
	\caption{Summary statistics of selected variables in the capital-market dataset}
	\footnotesize
	\begin{tabular}{l|ccccccc}
		\hline
		\textbf{Series} & \textbf{Mean} & \textbf{SD} & \textbf{Median} & 
		\textbf{Min} & \textbf{Max} & \textbf{Skew} & \textbf{Kurtosis} \\ \hline
		SET         & 0.0367   & 1.059   & 0.012    & $-11.43$       & 7.65      & $-0.923$ & 14.91 \\
		MAI         & 0.0368   & 1.102   & 0.044    & $-8.00$        & 8.05      & $-1.070$ & 11.51 \\
		ZeroShort   & $-0.0004$ & 0.013   & 0.000    & $-0.25$        & 0.09      & $-5.679$ & 99.13 \\
		ZeroLong    & $-0.0003$ & 0.027   & 0.000    & $-0.21$        & 0.10      & 0.652    & 13.96 \\
		CorpBond    & $-0.0001$ & 0.011   & 0.000    & $-0.12$        & 0.10      & $-0.008$ & 19.58 \\
		THB         & 0.0015    & 0.286   & 0.010    & $-1.56$        & 1.37      & $-0.068$ & 4.87 \\
		EquityFlow  & $-7.2057$ & 66.876  & $-2.42$  & $-543.16$      & 826.99    & 0.584    & 16.92 \\
		BondFlow    & 1,082.91  & 4,571.04 & 154.96   & $-26,973.48$  & 37,136.74 & 1.553    & 10.40 \\
		EMAsiaEquity & 0.0384   & 1.107   & 0.070    & $-5.87$        & 5.62      & $-0.341$ & 6.24 \\
		SP500       & 0.0483    & 1.125   & 0.050    & $-12.77$       & 8.97      & $-0.692$ & 16.76 \\
		EMBond      & 0.0289    & 0.252   & 0.030    & $-3.71$        & 2.02      & $-2.538$ & 37.42 \\
		USBond      & 0.0164    & 0.198   & 0.020    & $-2.19$        & 1.25      & $-0.735$ & 10.59 \\
		EMAsiaFX    & 0.0005    & 0.234   & 0.010    & $-1.64$        & 1.32      & $-0.219$ & 5.74 \\
		USD         & 0.0031    & 0.397   & 0.000    & $-2.51$        & 2.26      & 0.038     & 5.29 \\ \hline
	\end{tabular}
	\label{tab:desc-stats-full}
\end{table}

To assess model performance, we fitted the R-, C-, D-, Gaussian-, SL-, and $\Delta$-vine copulas. As shown in Table~\ref{monggon:tabThai}, the R-vine, SL-vine, and $\Delta$-vine consistently rank among the top three models under all evaluation criteria. Based on the log-likelihood, the SL-vine provides the best fit. According to the AIC, the R-vine ranks first, followed by the SL-vine and $\Delta$-vine. In contrast, the BIC selects the $\Delta$-vine as the best model, followed by the R-vine, while the SL-vine yields the largest BIC value among the three. The corresponding structures differ, as illustrated in Figure~\ref{TOP3}, which displays the first tree $T_1$ of these top-performing models.

\begin{table}[htbp]
	\centering
	\caption{Comparison of copula models for capital-market dataset based on Log-Likelihood, AIC, and BIC.}
	\small
	\begin{tabular}{l|ccc}
		\hline
		Copula model & Log-likelihood & AIC & BIC \\ \hline
		R-vine copula        & $8,038.15$ & {\color{red}$-15,790.31$} & $-14,913.96$ \\ 
		$\Delta$-vine copula & $8,034.66$ & $-15,787.31$ & {\color{red}$-14,923.22$} \\ 
		SL-vine copula       & {\color{red}$8,045.07$} & $-15,790.13$ & $-14,870.89$ \\ 
		C-vine copula        & $7,969.66$ & $-15,653.32$ & $-14,776.98$ \\ 
		D-vine copula        & $7,900.42$ & $-15,486.85$ & $-14,524.71$ \\ 
		Gaussian vine copula & $6,615.55$ & $-13,049.10$ & $-12,491.43$ \\ \hline
	\end{tabular}
	\label{monggon:tabThai}
\end{table}

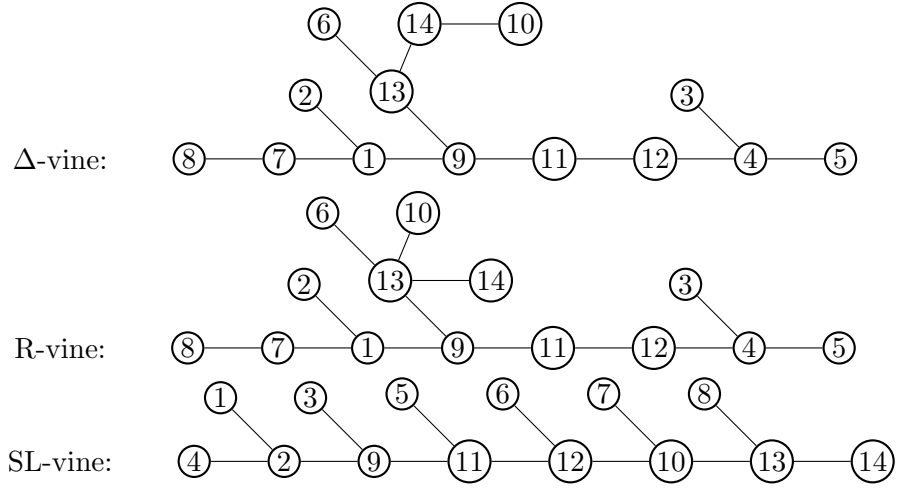
\begin{figure}[htbp]
\centering
\begin{tikzpicture}[every node/.style={thick, inner sep=1pt, font=\normalsize}, node distance=0.75 cm]
    
    % T1 layer
    \node (T1) at (-2,0) {$\Delta$-vine:};
    
    \node[circle, draw] (2) [right=2 cm of T1] {$7$};
    \node[circle, draw] (1) [left=of 2] {$8$};
    \node[circle, draw] (4) [right=of 2] {$1$};
    \node[circle, draw] (5) [right=of 4] {$9$};
    \node[circle, draw] (6) [right=of 5] {$11$};
     \node[circle, draw] (9) [right=of 6] {$12$};
     \node[circle, draw] (10) [right=of 9] {$4$};
     \node[circle, draw] (11) [right=of 10] {$5$};

    \node[circle, draw] (7) [above left =of 5] {$13$};
    \node[circle, draw] (8) [above left =of 4] {$2$};
    \node[circle, draw] (12) [above left =of 10] {$3$};
    \node[circle, draw] (x) [above left =of 7] {$6$};
    \node[circle, draw] (y) [ right =of x] {$14$};
    \node[circle, draw] (z) [right =of y] {$10$};

    \draw (1) -- (2);
    \draw (2) -- (4);
    \draw (4) -- (5);
    \draw (5) -- (6);
    \draw (6) -- (9);
    \draw (9) -- (10);
    \draw (10) -- (11);
    \draw (5) -- (7);
    \draw (8) --  (4);
    \draw (10) --  (12);
    \draw (x) --  (7);
    \draw (y) --  (7);
    \draw (z) --  (y);

    % T1 layer
    \node (T1) at (-2,-2.5) {R-vine:};
    
    \node[circle, draw] (2) [right=2 cm of T1] {$7$};
    \node[circle, draw] (1) [left=of 2] {$8$};
    \node[circle, draw] (4) [right=of 2] {$1$};
    \node[circle, draw] (5) [right=of 4] {$9$};
    \node[circle, draw] (6) [right=of 5] {$11$};
     \node[circle, draw] (9) [right=of 6] {$12$};
     \node[circle, draw] (10) [right=of 9] {$4$};
     \node[circle, draw] (11) [right=of 10] {$5$};

    \node[circle, draw] (7) [above left =of 5] {$13$};
    \node[circle, draw] (8) [above left =of 4] {$2$};
    \node[circle, draw] (12) [above left =of 10] {$3$};
    \node[circle, draw] (x) [above left =of 7] {$6$};
    \node[circle, draw] (y) [ right =of x] {$10$};
    \node[circle, draw] (z) [right =of 7] {$14$};

    \draw (1) -- (2);
    \draw (2) -- (4);
    \draw (4) -- (5);
    \draw (5) -- (6);
    \draw (6) -- (9);
    \draw (9) -- (10);
    \draw (10) -- (11);
    \draw (5) -- (7);
    \draw (8) --  (4);
    \draw (10) --  (12);
    \draw (x) --  (7);
    \draw (y) --  (7);
    \draw (z) --  (7);
    
    % T1 layer
    \node (T1) at (-2,-4) {SL-vine:};
    
    \node[circle, draw] (2) [right=2 cm of T1] {$2$};
    \node[circle, draw] (1) [left=of 2] {$4$};
    \node[circle, draw] (4) [right=of 2] {$9$};
    \node[circle, draw] (5) [right=of 4] {$11$};
    \node[circle, draw] (6) [right=of 5] {$12$};
     \node[circle, draw] (9) [right=of 6] {$10$};
     \node[circle, draw] (10) [right=of 9] {$13$};
     \node[circle, draw] (11) [right=of 10] {$14$};

    \node[circle, draw] (7) [above left =of 5] {$5$};
    \node[circle, draw] (8) [above left =of 4] {$3$};
    \node[circle, draw] (3) [above left =of 2] {$1$};
    \node[circle, draw] (12) [above left =of 10] {$8$};
    \node[circle, draw] (13) [above left =of 9] {$7$};
    \node[circle, draw] (14) [above left =of 6] {$6$};

    \draw (1) -- (2);
    \draw (2) -- (4);
    \draw (4) -- (5);
    \draw (5) -- (6);
    \draw (6) -- (9);
    \draw (9) -- (10);
    \draw (10) -- (11);
    \draw (5) -- (7);
    \draw (8) --  (4);
    \draw (3) --  (2);
    \draw (6) --  (14);
    \draw (9) --  (13);
    \draw (10) --  (12);

\end{tikzpicture}
\caption{The $T_1$ trees of the top three R-, $\Delta$-, and SL-vines, respectively.}\label{TOP3}
\end{figure}

\newpage
The BIC values exhibit a clearer separation among the models than the AIC values, which differ only slightly. To examine this further, we consider the AIC difference, defined as
$$\Delta_i = \text{AIC}_i - \text{AIC}_{\min},$$
where $i$ denotes the vine type and $\text{AIC}_{\min}$ is the minimum AIC value. According to \cite{burnham2004multimodel}, values of $\Delta_i$ less than 4 indicate that the difference is not statistically significant. In this case, $\text{AIC}_{\min}=\text{AIC}_{\text{R}}$.            Figure~\ref{graph AIC} shows that $\Delta_{\text{SL}}=0.18<4$ and $\Delta_{\Delta}=3.00<4$. These differences are not statistically significant. Therefore, we rely on the BIC for the final model selection, which identifies the $\Delta$-vine copula as the best-fitting model for the capital-market dataset.

\begin{figure}[htbp]
\centering
\begin{tikzpicture}[scale=0.75]

% แกน
\draw[->, thick] (0,0) -- (0,3.5) node[above] {$\Delta_i$};
\draw[->, thick] (0,0) -- (3.5,0) node[right] {$i$};

% tick y-axis
\draw (0.1,0) -- (-0.1,0) node[left] {0.00};
\draw (0.1,1) -- (-0.1,1) node[left] {0.18};
\draw (0.1,3) -- (-0.1,3) node[left] {3.00};

% tick x-axis + labels
\node[below] at (1.5,0)  {SL};
\node[below] at (3,0)  {$\Delta$};
\node[below] at (0,0) {R};

% จุด scatter (ตำแหน่งโดยประมาณจากรูป)
\fill[blue] (0,0) circle (2pt);          % R
\fill[blue] (1.5,1) circle (2pt);       % SL
\fill[blue] (3,3) circle (2pt);        % Δ

\end{tikzpicture}
\caption{Comparison of $\Delta_i$ values across different vine structures (R-vine, SL-vine, and $\Delta$-vine).}\label{graph AIC}
\end{figure}
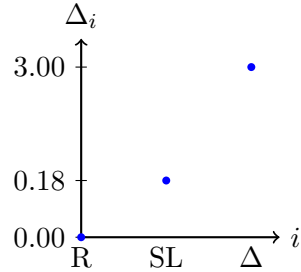

\subsection{Wisconsin breast cancer dataset}\label{SecBCW}
This section examines the Wisconsin Breast Cancer dataset from the University of California–Irvine Machine Learning Repository~(\url{https://doi.org/10.24432/C5DW2B}), which comprises diagnostic and prognostic subsets.
\subsubsection{Diagnostic group}
For the diagnostic subset, the dataset contains measurements extracted from digitized fine-needle aspirate (FNA) images of breast tissue and consists of 569 tumor samples classified as either malignant or benign. For each sample, ten morphological characteristics of the cell nuclei are recorded: radius, texture, perimeter, area, smoothness, compactness, concavity, concave points, symmetry, and fractal dimension. Each characteristic is summarized by its mean, standard deviation, and maximum value (defined as the mean of the three largest observations), yielding a total of 30 numerical variables. These features provide detailed descriptions of nuclear morphology and are commonly used for breast cancer diagnosis. In this study, only the mean values of the ten features are retained, reducing the dimensionality from 30 to 10 variables. Summary statistics for the Wisconsin breast cancer (diagnostic) dataset are reported in Table~\ref{tab:desc-stats-diabreastcancer}.

\begin{table}[htbp]
	\centering
	\caption{Descriptive statistics of diagnostic Wisconsin breast cancer dataset}
	\footnotesize
	\begin{tabular}{l|ccccccc}
		\hline
		\textbf{Feature} & \textbf{Mean} & \textbf{SD} & \textbf{Median} & 
		\textbf{Min} & \textbf{Max} & \textbf{Skew} & \textbf{Kurtosis} \\ \hline
		Radius             & 14.13  & 3.52   & 13.37  & 6.98   & 28.11   & 0.94 & 0.81 \\
		Texture            & 19.29  & 4.30   & 18.84  & 9.71   & 39.28   & 0.65 & 0.73 \\
		Perimeter          & 91.97  & 24.30  & 86.24  & 43.79  & 188.50  & 0.99 & 0.94 \\
		Area               & 654.89 & 351.91 & 551.10 & 143.50 & 2501.00 & 1.64 & 3.59 \\
		Smoothness         & 0.10   & 0.01   & 0.10   & 0.05   & 0.16    & 0.45 & 0.82 \\
		Compactness        & 0.10   & 0.05   & 0.09   & 0.02   & 0.35    & 1.18 & 1.61 \\
		Concavity          & 0.09   & 0.08   & 0.06   & 0.00   & 0.43    & 1.39 & 1.95 \\
		Concave points     & 0.05   & 0.04   & 0.03   & 0.00   & 0.20    & 1.17 & 1.03 \\
		Symmetry           & 0.18   & 0.03   & 0.18   & 0.11   & 0.30    & 0.72 & 1.25 \\
		Fractal dimension  & 0.06   & 0.01   & 0.06   & 0.05   & 0.10    & 1.30 & 2.95 \\ \hline
	\end{tabular}
	\label{tab:desc-stats-diabreastcancer}
\end{table}

To assess how well each model captures the dependence structure among these features, we compare several vine copula models using the same procedure as in the previous sections. Figure~\ref{TOP3dia} displays the first tree $T_1$ of the $\Delta$-, R-, and SL-vine models. The $\Delta$- and R-vines produce identical $T_1$ structures, with the $\Delta$-vine additionally imposing a node-degree constraint. This is consistent with the fact that both models attain identical log-likelihood, AIC, and BIC values. Furthermore, all three trees have node degrees of at most three, suggesting that the dependence structure of this dataset can be adequately represented by vine copulas with degree-constrained tree structures.

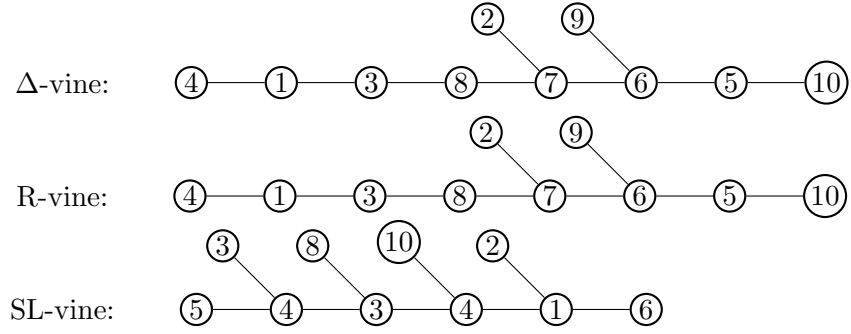
\begin{figure}[htbp]
\centering
\begin{tikzpicture}[every node/.style={thick, inner sep=1pt, font=\normalsize}, node distance=0.75 cm]

    % T1 layer
    \node (T1) at (-2,0) {$\Delta$-vine:};
    
    \node[circle, draw] (2) [right=2 cm of T1] {$1$};
    \node[circle, draw] (1) [left=of 2] {$4$};
    \node[circle, draw] (4) [right=of 2] {$3$};
    \node[circle, draw] (5) [right=of 4] {$8$};
    \node[circle, draw] (6) [right=of 5] {$7$};
    \node[circle, draw] (9) [right=of 6] {$6$};
    \node[circle, draw] (10) [right=of 9] {$5$};
    \node[circle, draw] (11) [right=of 10] {$10$};

    \node[circle, draw] (x) [above left=of 9] {$9$};
    \node[circle, draw] (y) [above left=of 6] {$2$};
     
    \draw (1) -- (2);
    \draw (2) -- (4);
    \draw (4) -- (5);
    \draw (5) -- (6);
    \draw (6) -- (9);
    \draw (9) -- (10);
    \draw (10) -- (11);
    \draw (x) -- (9);
    \draw (y) -- (6);
    
    % T1 layer
    \node (T1) at (-2,-1.5) {R-vine:};
    
    \node[circle, draw] (2) [right=2 cm of T1] {$1$};
    \node[circle, draw] (1) [left=of 2] {$4$};
    \node[circle, draw] (4) [right=of 2] {$3$};
    \node[circle, draw] (5) [right=of 4] {$8$};
    \node[circle, draw] (6) [right=of 5] {$7$};
    \node[circle, draw] (9) [right=of 6] {$6$};
    \node[circle, draw] (10) [right=of 9] {$5$};
    \node[circle, draw] (11) [right=of 10] {$10$};

    \node[circle, draw] (x) [above left=of 9] {$9$};
    \node[circle, draw] (y) [above left=of 6] {$2$};
     
    \draw (1) -- (2);
    \draw (2) -- (4);
    \draw (4) -- (5);
    \draw (5) -- (6);
    \draw (6) -- (9);
    \draw (9) -- (10);
    \draw (10) -- (11);
    \draw (x) -- (9);
    \draw (y) -- (6);

    % T1 layer
    \node (T1) at (-2,-3) {SL-vine:};
    
    \node[circle, draw] (2) [right=2 cm of T1] {$4$};
    \node[circle, draw] (1) [left=of 2] {$5$};
    \node[circle, draw] (4) [right=of 2] {$3$};
    \node[circle, draw] (5) [right=of 4] {$4$};
    \node[circle, draw] (6) [right=of 5] {$1$};
     \node[circle, draw] (9) [right=of 6] {$6$};

    \node[circle, draw] (7) [above left =of 5] {$10$};
    \node[circle, draw] (8) [above left =of 4] {$8$};
    \node[circle, draw] (3) [above left =of 2] {$3$};
    \node[circle, draw] (14) [above left =of 6] {$2$};

    \draw (1) -- (2);
    \draw (2) -- (4);
    \draw (4) -- (5);
    \draw (5) -- (6);
    \draw (6) -- (9);
    
    \draw (5) -- (7);
    \draw (8) --  (4);
    \draw (3) --  (2);
    \draw (6) -- (14);
\end{tikzpicture}
\caption{The $T_1$ trees of the top three $\Delta$-, R-,  and SL-vines, respectively.}\label{TOP3dia}
\end{figure} 

Table~\ref{monggon:tabdia} reports model comparisons based on log-likelihood, AIC, and BIC. The $\Delta$-vine and R-vine provide the best fit, achieving identical log-likelihood values and the lowest AIC and BIC scores. The SL-vine ranks third under both criteria but still outperforms the C-vine and D-vine copulas across all evaluation measures.

\begin{table}[htbp]
	\centering
	\caption{Comparison of copula models for diagnostic Wisconsin breast cancer dataset based on Log-Likelihood, AIC, and BIC.}
	\small
	\begin{tabular}{l|ccc}
		\hline
		Copula model & Log-likelihood & AIC & BIC \\ \hline
		$\Delta$-vine copula & $6,449.38$ & $-12,786.77$ & $-12,543.51$ \\
		R-vine copula       & $6,449.38$ & $-12,786.77$ & $-12,543.51$ \\ 
		SL-vine copula      & $6,442.43$ & $-12,754.86$ & $-12,472.51$ \\
		C-vine copula       & $6,334.35$ & $-12,544.70$ & $-12,275.38$ \\  
		D-vine copula       & $6,227.02$ & $-12,320.05$ & $-12,029.01$ \\ \hline
	\end{tabular}
	\label{monggon:tabdia}
\end{table}
\newpage
\subsubsection{Prognostic group}
To further assess the robustness of the results, we analyze the prognostic Wisconsin breast cancer dataset. The dataset contains 198 observations and 34 variables, including an identification number, clinical outcome variables, and measurements extracted from digitized fine-needle aspiration (FNA) images of breast tissue. The response variable classifies patients as recurrent or nonrecurrent, while the variable \textit{Time} records either the time to recurrence or the disease-free survival time, depending on the outcome.

As in the diagnostic dataset, the same ten morphological characteristics of the cell nuclei are considered. For each characteristic, the mean, standard deviation, and maximum value are recorded, yielding 30 quantitative variables. In this study, only the mean values are retained, reducing the dimensionality from 30 to 10 variables. Two additional clinical variables are included: tumor size (cm) and the number of positive axillary lymph nodes. Summary statistics for the prognostic Wisconsin breast cancer dataset are reported in Table~\ref{tab:desc-stats-pro}.

\begin{table}[htbp]
	\centering
	\caption{Descriptive statistics of prognostic Wisconsin breast cancer dataset}
	\footnotesize
	\begin{tabular}{l|ccccccc}
		\hline
		\textbf{Feature} & \textbf{Mean} & \textbf{SD} & \textbf{Median} & 
		\textbf{Min} & \textbf{Max} & \textbf{Skew} & \textbf{Kurtosis} \\ \hline
		Time              & 46.94  & 34.52  & 39.50  & 1.00   & 125.00  & 0.51 & -0.84 \\
		Radius            & 17.40  & 3.17   & 17.29  & 10.95  & 27.22   & 0.32 & -0.31 \\
		Texture           & 22.30  & 4.34   & 21.80  & 10.38  & 39.28   & 0.54 & -0.83 \\
		Perimeter         & 114.78 & 21.43  & 113.70 & 71.90  & 182.10  & 0.39 & -0.16 \\
		Area              & 969.09 & 353.16 & 929.10 & 361.60 & 2250.00 & 0.72 & 0.39 \\
		Smoothness        & 0.10   & 0.01   & 0.10   & 0.07   & 0.14    & 0.42 & 0.35 \\
		Compactness       & 0.14   & 0.05   & 0.13   & 0.05   & 0.31    & 0.60 & 0.38 \\
		Concavity         & 0.16   & 0.07   & 0.15   & 0.02   & 0.43    & 0.67 & 0.65 \\
		Concave points    & 0.09   & 0.03   & 0.09   & 0.02   & 0.20    & 0.69 & 0.66 \\
		Symmetry          & 0.19   & 0.03   & 0.19   & 0.13   & 0.30    & 0.75 & 0.91 \\
		Fractal dimension & 0.06   & 0.01   & 0.06   & 0.05   & 0.10    & 0.97 & 1.71 \\
		Tumor size        & 2.87   & 1.95   & 2.50   & 0.40   & 10.00   & 1.70 & 2.96 \\
		Lymph node status & 3.21   & 5.48   & 1.00   & 0.00   & 27.00   & 2.22 & 4.75 \\ \hline
	\end{tabular}
	\label{tab:desc-stats-pro}
\end{table}

Table~\ref{monggon:tabpro} compares model performance in capturing the dependence structure among the variables. The $\Delta$-vine and R-vine copulas achieve identical log-likelihood, AIC, and BIC values and provide the best fit among all models considered. The SL-vine ranks third and continues to outperform the C-vine and D-vine copulas.

\begin{table}[htbp]
	\centering
	\caption{Comparison of copula models for prognostic wisconsin breast cancer dataset based on Log-Likelihood, AIC, and BIC.}
	\small
	\begin{tabular}{l|ccc}
		\hline
		Copula model & Log-likelihood & AIC & BIC \\ \hline
		$\Delta$-vine copula & $2,100.30$ & $-4,018.61$ & $-3,721.23$ \\
		R-vine copula       & $2,100.30$ & $-4,018.61$ & $-3,721.23$ \\ 
		SL-vine copula      & $2,078.23$ & $-3,998.46$ & $-3,740.30$ \\
		C-vine copula       & $2,083.89$ & $-3,981.79$ & $-3,677.88$ \\  
		D-vine copula       & $1,965.62$ & $-3,777.23$ & $-3,525.61$ \\ \hline
	\end{tabular}
	\label{monggon:tabpro}
\end{table}

Figure~\ref{TOP3pro} shows the first trees $T_1$ of the $\Delta$-, R-, and SL-vine models. As in the Diagnostic dataset, the $\Delta$-vine and R-vine produce identical first-tree structures, consistent with their identical evaluation criteria. In contrast, the SL-vine exhibits a more branched structure while maintaining the node-degree constraint. Furthermore, all three trees have node degrees not exceeding three, indicating that the dependence structure of this dataset can be adequately represented by vine copulas with a maximum node degree of three. These results further support the effectiveness of the $\Delta$-vine and R-vine relative to the other vine copula models considered.

\begin{figure}[htbp]
\centering
\begin{tikzpicture}[every node/.style={thick, inner sep=1pt, font=\normalsize}, node distance=0.75 cm]

    % T1 layer
    \node (T1) at (-2,0) {$\Delta$-vine:};
    
    \node[circle, draw] (2) [right=2 cm of T1] {$11$};
    \node[circle, draw] (1) [left=of 2] {$6$};
    \node[circle, draw] (4) [right=of 2] {$7$};
    \node[circle, draw] (5) [right=of 4] {$8$};
    \node[circle, draw] (6) [right=of 5] {$9$};
    \node[circle, draw] (9) [right=of 6] {$4$};
    \node[circle, draw] (10) [right=of 9] {$2$};
    \node[circle, draw] (11) [right=of 10] {$5$};
    \node[circle, draw] (12) [right=of 11] {$12$};
    \node[circle, draw] (13) [above left=of 12] {$13$};

    \node[circle, draw] (x) [above left=of 4] {$10$};
    \node[circle, draw] (z) [above left=of 6] {$3$};
    \node[circle, draw] (y) [above left=of 9] {$1$};
     
    \draw (1) -- (2);
    \draw (2) -- (4);
    \draw (4) -- (5);
    \draw (5) -- (6);
    \draw (6) -- (9);
    \draw (9) -- (10);
    \draw (10) -- (11);
    \draw (11) -- (12);
    \draw (12) -- (13);
    \draw (x) -- (4);
    \draw (y) -- (9);
    \draw (z) -- (y);
    
    % T1 layer
    \node (T1) at (-2,-2) {R-vine:};
    
    \node[circle, draw] (2) [right=2 cm of T1] {$11$};
    \node[circle, draw] (1) [left=of 2] {$6$};
    \node[circle, draw] (4) [right=of 2] {$7$};
    \node[circle, draw] (5) [right=of 4] {$8$};
    \node[circle, draw] (6) [right=of 5] {$9$};
    \node[circle, draw] (9) [right=of 6] {$4$};
    \node[circle, draw] (10) [right=of 9] {$2$};
    \node[circle, draw] (11) [right=of 10] {$5$};
    \node[circle, draw] (12) [right=of 11] {$12$};
    \node[circle, draw] (13) [above left=of 12] {$13$};

    \node[circle, draw] (x) [above left=of 4] {$10$};
    \node[circle, draw] (z) [above left=of 6] {$3$};
    \node[circle, draw] (y) [above left=of 9] {$1$};
     
    \draw (1) -- (2);
    \draw (2) -- (4);
    \draw (4) -- (5);
    \draw (5) -- (6);
    \draw (6) -- (9);
    \draw (9) -- (10);
    \draw (10) -- (11);
    \draw (11) -- (12);
    \draw (12) -- (13);
    \draw (x) -- (4);
    \draw (y) -- (9);
    \draw (z) -- (y);
    
    % T1 layer
    \node (T1) at (-2,-4) {SL-vine:};

    \node[circle, draw] (2) [right=2 cm of T1] {$5$};
    \node[circle, draw] (1) [left=of 2] {$2$};
    \node[circle, draw] (4) [right=of 2] {$4$};
    \node[circle, draw] (5) [right=of 4] {$9$};
    \node[circle, draw] (6) [right=of 5] {$8$};
     \node[circle, draw] (9) [right=of 6] {$7$};
     \node[circle, draw] (y) [right=of 9] {$11$};

    \node[circle, draw] (7) [above left =of 5] {$3$};
    \node[circle, draw] (8) [above left =of 4] {$6$};
    \node[circle, draw] (3) [above left =of 2] {$10$};
    \node[circle, draw] (14) [above left =of 6] {$12$};
    \node[circle, draw] (x) [above left =of 9] {$1$};
    \node[circle, draw] (z) [above left =of y] {$13$};

    \draw (1) -- (2);
    \draw (2) -- (4);
    \draw (4) -- (5);
    \draw (5) -- (6);
    \draw (6) -- (9);
    
    \draw (5) -- (7);
    \draw (8) --  (4);
    \draw (3) --  (2);
    \draw (6) -- (14);
    \draw (y) -- (9);
    \draw (x) -- (9);
    \draw (z) -- (y);
\end{tikzpicture}
\caption{The $T_1$ trees of the top three $\Delta$-, R-,  and SL-vines, respectively.}\label{TOP3pro}
\end{figure}
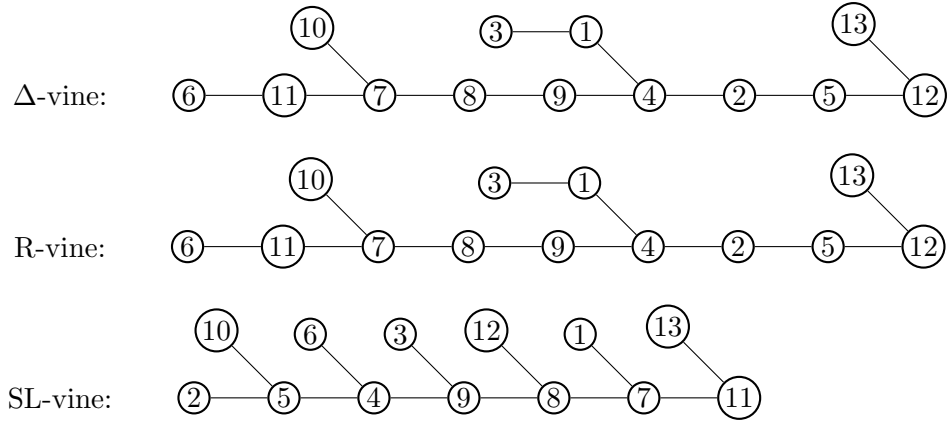 

\section{Conclusions and further research}\label{sec7}
In this paper, we propose modified vine copula models with degree-constrained tree structures for modeling multivariate dependence. We first introduce the stem-and-leaf (SL-) vine, a hybrid structure combining features of C-vines and D-vines. The SL-vine can capture both hierarchical and sequential dependencies while restricting node degrees to at most three. We establish its theoretical validity by showing that it satisfies the R-vine tree sequence conditions and possesses a well-defined joint density function. An algorithm for constructing the SL-vine is also developed for practical implementation. To overcome the structural rigidity of the SL-vine, we further introduce the degree-of-at-most-three vine, or $\Delta$-vine, as a generalization of the SL-vine. The $\Delta$-vine remains within the R-vine framework, allowing existing theoretical results and computational tools to be directly applied. Its cycle-elimination mechanism provides a systematic and interpretable procedure for automated vine construction while preserving flexibility in the tree structure. In addition, we demonstrate how SL- and $\Delta$-vine copulas can be simulated within the R-vine framework, providing practical guidance for generating synthetic multivariate data from the proposed models.

The proposed models are evaluated using three benchmark datasets: abalone, capital-market, and Wisconsin breast cancer. For the abalone dataset, the SL-vine achieves the best performance among the competing models. For the capital-market dataset, the $\Delta$-vine attains the best overall performance, although its improvement over the SL- and R-vines is not statistically significant. For both the diagnostic and prognostic breast cancer datasets, the $\Delta$-vine and R-vine yield identical best-fitting models. Overall, the results demonstrate that the SL- and $\Delta$-vine copulas provide competitive and effective alternatives to existing vine copula models for representing complex dependence structures.

An interesting observation is that the more general R-vine does not always outperform the proposed models. Although the R-vine allows a broader class of tree structures, its sequential maximum-spanning-tree construction does not necessarily explore all structures with maximal overall dependence. Consequently, the resulting R-vine may not coincide with the optimal SL- or $\Delta$-vine configuration. This suggests several directions for future research, including identifying conditions under which particular vine structures are preferable and developing selection procedures that account for both model flexibility and structural constraints.

\section*{Acknowledgement}
The author appreciates the referees for their thorough review of the manuscript and their valuable feedback. The His Royal Highness Crown Prince Maha Vajiralongkorn Scholarship from the Graduate School, Chulalongkorn University to commemorate the 72nd anniversary of his Majesty King Bhumibala Aduladeja is gratefully acknowledged.

\section*{Funding}
The authors did not receive support from any organization for the submitted work.

%\section*{Notes on contributor(s)}

%\section*{Nomenclature/Notation}

%\section*{Notes}

\bibliographystyle{tfs}
\bibliography{reference.bib}

@inproceedings{sklar1959fonctions,
  title={Fonctions de r{\'e}partition {\`a} n dimensions et leurs marges},
  author={Sklar, M},
  booktitle={Annales de l'ISUP},
  volume={8},
  number={3},
  pages={229--231},
  year={1959}
}

@article{jones2009formula,
  title={The formula that felled Wall St},
  author={Jones, Sam},
  journal={Financial Times},
  volume={24},
  pages={2009},
  year={2009}
}

@article{aas2009pair,
  title={Pair-copula constructions of multiple dependence},
  author={Aas, Kjersti and Czado, Claudia and Frigessi, Arnoldo and Bakken, Henrik},
  journal={Insurance: Mathematics and economics},
  volume={44},
  number={2},
  pages={182--198},
  year={2009},
  publisher={Elsevier}
}

@article{joe1996families,
  title={Families of m-variate distributions with given margins and m (m-1)/2 bivariate dependence parameters},
  author={Joe, Harry},
  journal={Lecture notes-monograph series},
  pages={120--141},
  year={1996},
  publisher={JSTOR}
}

@article{bedford2001probability,
  title={Probability density decomposition for conditionally dependent random variables modeled by vines},
  author={Bedford, Tim and Cooke, Roger M},
  journal={Annals of Mathematics and Artificial intelligence},
  volume={32},
  number={1},
  pages={245--268},
  year={2001},
  publisher={Springer}
}

@article{beare2015vine,
  title={Vine copula specifications for stationary multivariate Markov chains},
  author={Beare, Brendan K and Seo, Juwon},
  journal={Journal of Time Series Analysis},
  volume={36},
  number={2},
  pages={228--246},
  year={2015},
  publisher={Wiley Online Library}
}

@book{kurowicka2006uncertainty,
  title={Uncertainty analysis with high dimensional dependence modelling},
  author={Kurowicka, Dorota and Cooke, Roger M},
  year={2006},
  publisher={John Wiley \& Sons}
}

@article{HOBAEKHAFF20101296,
title = {On the simplified pair-copula construction — Simply useful or too simplistic?},
journal = {Journal of Multivariate Analysis},
volume = {101},
number = {5},
pages = {1296-1310},
year = {2010},
issn = {0047-259X},
doi = {https://doi.org/10.1016/j.jmva.2009.12.001},
url = {https://www.sciencedirect.com/science/article/pii/S0047259X09002206},
author = {Ingrid {Hobæk Haff} and Kjersti Aas and Arnoldo Frigessi}
}

@article{STOBER2013101,
title = {Simplified pair copula constructions—Limitations and extensions},
journal = {Journal of Multivariate Analysis},
volume = {119},
pages = {101-118},
year = {2013},
issn = {0047-259X},
doi = {https://doi.org/10.1016/j.jmva.2013.04.014},
url = {https://www.sciencedirect.com/science/article/pii/S0047259X13000754},
author = {Jakob Stöber and Harry Joe and Claudia Czado}
}

@article{dissmann2013selecting,
  title={Selecting and estimating regular vine copulae and application to financial returns},
  author={Dissmann, Jeffrey and Brechmann, Eike C and Czado, Claudia and Kurowicka, Dorota},
  journal={Computational Statistics \& Data Analysis},
  volume={59},
  pages={52--69},
  year={2013},
  publisher={Elsevier}
}

@article{czado2012maximum,
  title={Maximum likelihood estimation of mixed C-vines with application to exchange rates},
  author={Czado, Claudia and Schepsmeier, Ulf and Min, Aleksey},
  journal={Statistical Modelling},
  volume={12},
  number={3},
  pages={229--255},
  year={2012},
  publisher={SAGE Publications Sage India: New Delhi, India}
}

@article{brechmann2010truncated,
  title={Truncated and simplified regular vines and their applications},
  author={Brechmann, Eike},
  year={2010}
}

@article{kuhn1955hungarian,
  title={The Hungarian method for the assignment problem},
  author={Kuhn, Harold W},
  journal={Naval research logistics quarterly},
  volume={2},
  number={1-2},
  pages={83--97},
  year={1955},
  publisher={Wiley Online Library}
}

@article{nagler2024solving,
  title={Solving estimating equations with copulas},
  author={Nagler, Thomas and Vatter, Thibault},
  journal={Journal of the American Statistical Association},
  volume={119},
  number={546},
  pages={1168--1180},
  year={2024},
  publisher={Taylor \& Francis}
}

@article{Nacaskul2023MultiParadigm,
  title        = {Multi‐Paradigm Analysis of Thai Capital Market Linkages: Bivariate/Vine Copulas, Granger Causality, Network Centrality, and Graph Neural Network/Graph Embedding Approaches},
  author       = {Nacaskul, Poomjai and Kalakan, Kongkan and Khlaisamniang, Pitikorn and Veerabulyarith, Puvarith and Sukcharoenchaikul, Isariyaporn},
  year         = {2023},
  journal      = {Papers with Code},
  howpublished = {\url{https://paperswithcode.com/paper/multi-paradigm-analysis-of-thai-capital}}
}

@article{burnham2004multimodel,
  title={Multimodel inference: understanding AIC and BIC in model selection},
  author={Burnham, Kenneth P and Anderson, David R},
  journal={Sociological methods \& research},
  volume={33},
  number={2},
  pages={261--304},
  year={2004},
  publisher={Sage Publications Sage CA: Thousand Oaks, CA}
}

\appendix
\section{Di{\ss}mann’s algorithm}\label{app:Diss}
\begin{algorithm}[H]
\caption{Di{\ss}mann’s Algorithm for Vine Copula Model}
\label{alg:discmann1}
\begin{algorithmic}[1]
\State Compute the weights $\tau_{i,j}$ for all index pairs $\{i, j\}$, with $1 \leq i < j \leq n$.
\State Select the maximum spanning tree:
$$T_1 = \underset{
\substack{
T=(N,E)\ \text{spanning tree} 
}
}{\arg\max} \sum_{e = (a_e, b_e) \in E} \tau_{a_e, b_e}.$$
\For{each edge $e \in E_1$}
    \State Choose a copula $C_{a_e, b_e}$ with estimated parameters $\widehat{\theta}_{a_e, b_e}$.
    \For{$k = 1 \text{ to } n$}
        \State Compute pseudo-observations:
        \begin{align*}
          u_{k,a_e|b_e,\widehat{\theta}_{a_e,b_e}}&=C_{a_e|b_e}(u_{k,a_e} \mid u_{k,b_e}; \widehat{\theta}_{a_e,b_e}) \text{ and } \\
        u_{k,b_e|a_e,\widehat{\theta}_{a_e,b_e}}&=C_{b_e|a_e}(u_{k,b_e} \mid u_{k,a_e}; \widehat{\theta}_{a_e,b_e})  
        \end{align*}
    \EndFor
\EndFor
\For{$m = 2$ to $d-1$}
    \State Compute weights $\tau_{a_e,b_e \mid D_e}$ for all eligible edges $(a_e, b_e \mid D_e)$ in $T_m$.
    \State Let $E_{P,m}$ be the set of all edges satisfying the proximity condition.
    \State Select the maximum spanning tree: 
    $$T_m = \underset{
\substack{
T=(N,E)\ \text{spanning tree with} E \subset E_{P,m}
}
}{\arg\max}\sum_{e \in E} \tau_{a_e,b_e \mid D_e}$$
   
    \For{each edge $e \in E_m$}
        \State Choose a pair copula $C_{a_e,b_e \mid D_e}$ with estimated parameters $\widehat{\theta}_{a_e,b_e \mid D_e}$.
        \For{$k = 1 \text{ to } n$}
            \State Compute pseudo-observations:
            \begin{align*}
            u_{k,a_e|b_e\cup D_e,\widehat{\theta}_{a_e,b_e;D_e}}&=C_{a_e|b_e \cup D_e}(u_{k,a_e} \mid u_{k,b_e}, \mathbf{u}_{k,D_e}; \widehat{\theta}_{a_e,b_e;D_e})
                \text{ and }\\
            u_{k,b_e|a_e\cup D_e,\widehat{\theta}_{a_e,b_e;D_e}}&=C_{b_e|a_e \cup D_e}(u_{k,b_e} \mid u_{k,a_e}, \mathbf{u}_{k,D_e}; \widehat{\theta}_{a_e,b_e;D_e})
            \end{align*}
        \EndFor
    \EndFor
\EndFor
\State \Return the model estimates $\left(\widehat{\mathcal{V}}, \widehat{\mathcal{B}}, \widehat{\boldsymbol{\theta}}\right)$
\end{algorithmic}
\end{algorithm}

\section{Algorithm for SL-Vine}\label{app:AlgSL}
\begin{algorithm}[H]
\caption{GetProblem(\textit{GRAPH})}
\label{alg:getproblem}
\begin{algorithmic}[1]
\State Set \textit{Triangles} $\gets \emptyset$ \Comment{To contain all triangles in \textit{GRAPH}}
\State Set \textit{ProblemPoints} $\gets \emptyset$ \Comment{To contain list of all 4-degree points in each triangle}
\State Set \textit{Problems} $\gets \emptyset$ \Comment{To contain the list of \{chain of triangles, and 4-degree point that is not in between triangles\}}
\State Set \textit{Checked} $\gets \emptyset$ \Comment{To contain checked edge from following for loop}
\For{each \textit{CheckingEdge} in all edges of \textit{GRAPH} not in \textit{Checked}}
    \If{there is another point that forms a triangle with \textit{CheckingEdge} to be a triangle in \textit{GRAPH}}
        \State Add all edges in that triangle to \textit{Checked}
        
        \State Add the triangle to \textit{Triangles}
        \State Add list of 4-degree points in the triangle (or empty list) to \textit{ProblemPoints}
    \Else
        \State Add \textit{CheckingEdge} to \textit{Checked}
    \EndIf
\EndFor
\If{there is an empty list in \textit{ProblemPoints}}
    \For{each empty list}
        \State Add \{triangle in\textit{ Triangles}, empty element\} to \textit{Problems}
    \EndFor
    \State Remove added triangles from \textit{Triangles}
    \While{\textit{Triangles} is not empty}
        \State Find a triangle in \textit{Triangles} with 1 element in its corresponding list in \textit{ProblemPoints}
        \State Let \textit{RemainingPoint} be that 1 element
        \State Set a 3-column matrix \textit{TriangleChain} with the first row be that triangle
        \State Remove the triangle from \textit{Triangles}
        \While{there is a triangle \textit{NextTriangle} in \textit{Triangles} with \textit{RemainingPoint}}
            \State Add \textit{NextTriangle} to next row of \textit{TriangleChain}
            \State Remove \textit{NextTriangle} from \textit{Triangles}
            \State Let \textit{RemainingPoint} be another point in the corresponding list of \textit{NextTriangle} in   \textit{ProblemPoints} or an empty if there is not any another point
        \EndWhile
        \State Add list of \{\textit{TriangleChain}, \textit{RemainingPoint}\} to \textit{Problems}
    \EndWhile
\EndIf
\State \Return \textit{Problems}
\end{algorithmic}
\end{algorithm}

\begin{algorithm}[H]
\caption{SL-vine Selection}
\label{alg:slvineselection}
\begin{algorithmic}[1]
\While{GetProblem(\textit{GRAPH}) is not empty}
    \For{each \{ \textit{TriangleChain}, \textit{RemainingPoint} \} in GetProblem(\textit{GRAPH})}
        \State Set the edge collection \textit{Pool} $\gets \emptyset$ \Comment{Pool for possible-to-delete edges}
        \If{\textit{RemainingPoint} is an empty element}
            \State Add all edges in \textit{TriangleChain} to \textit{Pool}
        \Else
            \State Add edges in \textit{TriangleChain} connected with \textit{RemainingPoint} to \textit{Pool}
        \EndIf
        \State Remove the edge with $\min{|\tau|}$ in \textit{Pool} from \textit{GRAPH}
    \EndFor
\EndWhile
\State \Return \textit{GRAPH}
\end{algorithmic}
\end{algorithm}

\section{Algorithm for $\Delta$-Vine}\label{app:AlgDel}
\begin{algorithm}[H]
\caption{$\Delta$-vine Selection}
\label{GetToDelete}
\begin{algorithmic}[1]
\State Set \textit{ToDelete} $\gets \emptyset$ \Comment{To contain all edges to be deleted from the line graph, making the graph to be in $\Delta$-vine structure}
\State Set \textit{Triangles} $\gets \emptyset$ \Comment{To contain all triangles in \textit{GRAPH}}
\State Set \textit{CheckedEdge} $\gets \emptyset$ \Comment{To contain checked edge from triangle finding process (For loop)}
\State Set $D_k \gets \emptyset$ \Comment{To contain checked triangles in the elimination process (While loop)}
\For{each \textit{CheckingEdge} in all edges of \textit{GRAPH} not in \textit{CheckedEdge}}
    \If{there exists a vertex $u$ such that $u$ forms a triangle with the endpoints of 
    \textit{CheckingEdge} in \textit{GRAPH}}
        \State Add all edges in that triangle to \textit{CheckedEdge}
    \Else
        \State Add \textit{CheckingEdge} to \textit{CheckedEdge}
    \EndIf
\EndFor
\While{\textit{Triangles} is not empty}
    \State Set point collection $V \gets$ intersection of \textit{Triangles} and \textit{CheckedTriangle}
    \If{$V = \emptyset$}
        \State Find edge $e_k$ with $\min{|\tau|}$ from all edges of all triangles in \textit{Triangles}
        \State Set $C_k \gets$ triangle in \textit{Triangles} that contains $e_k$
    \Else
        \State Select one point $v$ in $V$
        \State Set $C_k \gets$ triangle in \textit{Triangles} that contains $v$
        \State Set \textit{Pool} $\gets \emptyset$ \Comment{To contain possible-to-delete edges}
        \If{\textit{ToDelete} has an edge containing $v$}
            \State Find edge $e_k$ with $\min{|\tau|}$ from all edges in $C_k$ connected with $v$
        \Else
            \State Find edge $e_k$ with $\min{|\tau|}$ from all edges in $C_k$
        \EndIf
    \EndIf
    \State Add $e_k$ to \textit{ToDelete}
    \State Add $C_k$ to $D_k$
    \State Remove $C_k$ from \textit{Triangles}
\EndWhile
\State Remove all edges in \textit{ToDelete} from \textit{GRAPH}
\State \Return \textit{GRAPH}
\end{algorithmic}
\end{algorithm}

\end{document}